\documentclass[11pt, a4paper,reqno]{amsart}
\usepackage{amsfonts}
\usepackage{amssymb}
\usepackage{mathrsfs}
\usepackage[dvips]{graphics}
\usepackage{epsfig}
\usepackage{euscript}
\usepackage{color}

\usepackage{fullpage,hyperref}

\usepackage{fancyvrb}

 \newtheorem{thm}{Theorem}[section]
 \newtheorem{cor}[thm]{Corollary}
 \newtheorem{lemma}[thm]{Lemma}
 \newtheorem{prop}[thm]{Proposition}
 \theoremstyle{definition}
 \newtheorem{defn}[thm]{Definition}
 \newtheorem{rem}[thm]{Remark}
 \newtheorem{ex}{Example}
 \newtheorem{assumption}[thm]{Assumption}
 \numberwithin{equation}{section}

\newcommand{\R}{\mathbb{R}}
\newcommand{\N}{\mathbb{N}}
\newcommand{\C}{\mathbb{C}}                           

\newcommand{\Z}{\mathbb{Z}}
\newcommand{ \ii}{\,\mathrm{i}\,}

\newcommand{\s}[1]{\CMcal{#1}}
                  
\newcommand{\bb}[1]{\mathscr{#1}}
\newcommand{\rr}[1]{\mathfrak{#1}}

\newcommand{\expo}[1]{\,\mathrm{e}^{#1}\,}            
\newcommand{\ketbra}[2]{|#1\rangle\langle#2|}

\newcommand{\caA}{{\mathcal A}}

\newcommand{\caD}{{\mathcal D}}

\newcommand{\caH}{{\mathcal H}}

\newcommand{\caN}{{\mathcal N}}

\newcommand{\bbG}{{\mathbb G}}
\newcommand{\bbH}{{\mathbb H}}

\newcommand{\bbL}{{\mathbb L}}

\newcommand{\bbN}{{\mathbb N}}

\newcommand{\bbR}{{\mathbb R}}

\newcommand{\bbZ}{{\mathbb Z}}

\newcommand{\ie}{{\it i.e.\/} }

\newcommand{\iu}{\mathrm{i}}

\newcommand{\str}{^{*}}
\newcommand{\ep}[1]{\mathrm{e}^{#1}}

\newcommand{\dd}{\,\mathrm{d}}

\newcommand{\Tr}{\mathrm{Tr}}

\begin{document}
\title{Lieb-Robinson bounds in the continuum via localized frames}

\author[1]{Sven Bachmann}
\author[2]{Giuseppe De~Nittis$^\star$}

\address[1]{Department of Mathematics, University of British Columbia, Vancouver, BC V6T 1Z2, Canada}
\email{sbach@math.ubc.ca}
\address[2]{Facultad de Matem\'aticas \& Instituto de F\'{\i}sica,
  Pontificia Universidad Cat\'olica de Chile,
  Santiago, Chile.}
\email{gidenittis@mat.uc.cl}

\date{\today\\ $^\star$Corresponding author}                 
 
\begin{abstract}
We study the dynamics of interacting fermions in the continuum. Our approach uses the concept of lattice-localized frames, which we introduce here. We first prove a Lieb-Robinson bound that is valid for a general class of local interactions, which implies the existence of the dynamics at the level of the CAR algebra. We then turn to the physical situation relevant to the (fractional) quantum Hall effect, namely the quasi-free second quantized Landau Hamiltonian to which electron-electron interactions can be added.
\end{abstract}

\maketitle


\section{Introduction}   \label{sec: introduction}

The quantum dynamics of electrons in condensed matter systems in the thermodynamic limit is a phenomenologically very rich and theoretically very well understood topics. The integer quantum Hall effect is one of the most remarkable effects that have been observed and is by now mathematically well understood~\cite{TKNN,Laughlin,AvronSeiler85,AvronSeilerSimonCMP94}. In the case of interacting electrons, integer quantization of the Hall conductance can be proved for lattice systems in large but finite volume~\cite{HastingsMichalakis,OurHall} and in the infinite volume limit~\cite{kapustin2020hall}. In the fractional case where interactions are necessary, progress was made in~\cite{RationalQH,AbelianAnyons}. The key ingredient to all of these works is the Lieb-Robinson bound~\cite{Lieb:1972ts}, which is a ballistic upper bound for the quantum dynamics.

Among the first and arguably most fundamental corollary of the Lieb-Robinson bound is that it allows for the construction of the quantum dynamics in the infinite volume limit~\cite{Robinson_LR,nachtergaele2006propagation}. By this, we mean here that in a system with sufficiently local interactions, there is a one-parameter family of $\ast$-automorphisms of the algebra of observables, namely the algebra of canonical anticommutation relations over the one-particle Hilbert space, e.g. $\ell^2(\mathbb{Z}^\nu)$ for spinless fermions on the regular lattice. It is worth pointing out immediately that the question of the unitary implementation of the dynamics in a given representation of the algebra in Hilbert space is a separate question which we will not address here --- and does not play an important role in this setting. 

The situation is completely different if the lattice $\mathbb{Z}^\nu$ is replaced by the continuum $\mathbb{R}^\nu$, which is arguably physically at least as relevant. Beyond the free case known as Bogoliubov automorphisms~\cite{Lundberg,Araki}, the only classical results we are aware of are~\cite{streater1968certain,streater1970time} for fermionic field theories. Notable progress was made recently, starting with~\cite{Gebert_LR} and complemented by~\cite{Marius1}. All of these works start with the observation that the ultraviolet problem is the key hurdle for a proper definition of the dynamics and resolve it by considering suitably smeared creation and annihilation operators to define the interaction. \cite{Gebert_LR,Marius1} construct the dynamics using the method that is standard in lattice systems, namely by first proving a finite volume Lieb-Robinson bound with parameters that are independent of the volume.

In this work, we adopt a parallel point of view on the same problem which is motivated by the situation of the quantum Hall effect. There (and without any driving) the perpendicular magnetic field $B$ induces cyclotron orbits of the otherwise free electrons. The electronic wavefunctions are localized in a Gaussian way around the centre of their cyclotron orbits: This provides a natural smearing at the scale of the magnetic length $\ell_B$. The interactions can therefore be expressed as being between these smeared out electrons. The cyclotron orbits are labelled by $\gamma\in\mathbb{R}^2$ but it is well-known that the corresponding wavefunctions $\chi_\gamma$ form an overcomplete basis, in direct analogy with classical coherent states of the harmonic oscillator. We leverage on this to extract a countable number of them that are labelled by a set $\Gamma$, effectively reducing the continuum to a lattice. The remaining set $\{\chi_\gamma:\gamma\in\Gamma\}$ is a frame~\cite{boon-zak-78,bargmann-butera-girardello-klauder-71,perelomov-71}, which is still overcomplete and in particular not orthonormal. As a consequence, the corresponding creation and annihilation operators do not exactly anticommute as would lattice creation and annihilation operators. Nonetheless, any observable can be expressed in terms of the operators $\{a(\chi_\gamma),a^*(\chi_\gamma):\gamma\in\Gamma\}$ and lattice techniques apply with appropriate adaptations. 

With this `frame based second quantization', we consider a large class of local interactions and prove by a variation on the standard techniques that the corresponding finite-volume dynamics satisfies a Lieb-Robinson bound. As a corollary, we prove convergence to a proper dynamics on the infinite volume algebra. These interactions are reminiscent of Haldane's pseudo-potentials~\cite{HaldanePot}, although this here is much more general as our goal is not to relate these interactions specifically to the fractional quantum Hall effect. We emphasize that nothing is lost in going from the continuum to the lattice since the latter is a frame. In particular, any interaction written in terms of the frame creation-annihilation operators could be expressed in the standard terms of second quantization. But the use of the frame allows for a rather straightforward proof of the propagation estimate. We note that~\cite{BrunoSimoneAmanda,SimoneAmanda} use an approach that similarly conjure `lattice techniques' for the continuum, in this case to prove spectral gaps for a $\nu$= 1/3 fractional quantum Hall system.

While we highlight the example of the magnetic frames, the techniques we use here would equally apply to another physically relevant situation, namely that of Bloch electrons with a spectral gap. It is known that the subspace associated with the Fermi projection in the gap can be spanned by a frame of localized states, the so-called `generalized Wannier functions'. Such a frame is indeed over-complete if and only if the Fermi projection is not topologically trivial (non-zero Chern number) and is an orthonormal basis in the case of a trivial topology, see~\cite{cornean2019parseval,auckly2018parseval}. This localized frame fits our assumptions and results in the validity of a Lieb-Robinson bound for the dynamics of interacting Bloch electrons in the continuum in a gapped regime.

The paper is organized as follows. We review fundamental facts about frames in Section~\ref{sec: Gabor frames}, and we introduce the new notion of a lattice-localized frame, which is motivated by the Landau Hamiltonian. The proof of a Lieb-Robinson bound and its consequence follows in Section~\ref{sec: LRB}. In Section~\ref{sec: Fock quantization}, we consider the case of a free dynamics in the Fock space associated with a pure, gauge-invariant quasi-free state, expressing in particular the free Hamiltonian in terms of the frame creation-annihilation operators. We thereby build a bridge  between the $C^*$-algebraic setting based on the notion of interactions and the second quantization based on the magnetic Gabor frame. Finally, Section~\ref{sec: Laudau Hamiltonian} reviews the construction of the magnetic frame associated with the Landau Hamiltonian, and proves that it satisfies the assumptions of the previous sections. Finally, a couple of technical proofs are presented in Section~\ref{sec:add_proof}. 

\subsection*{Acknowledgements} GD's research is supported by the grant \emph{Fondecyt Regular - 1230032}. SD acknowledges financial support from NSERC of Canada.


\section{Localized frames}   \label{sec: Gabor frames}
In this section, which is mainly based on \cite{christensen-03}, we present basic facts about the theory of frames and discuss the properties of a certain class of relevant frames, which we call \emph{lattice-localized frames}.

\subsection{Basis and Riesz basis}   \label{sub: frames}

Let $\caH$ be a separable Hilbert space and $\Gamma$ a discrete and countable metric space endowed with a distance $d:\Gamma\times \Gamma\to [0,+\infty]$. We will assume that 
\begin{equation}\label{eq:bound_cond}
m_\epsilon := \sup_{\gamma\in\Gamma}\sum_{\xi\in\Gamma}\expo{-\epsilon d(\gamma,\xi)}\;<+\infty\;,
\end{equation}
for all $\epsilon>0$. Note that $m_\epsilon\geqslant1$ for every $\epsilon>0$ and the function $\epsilon\mapsto m_\epsilon$ is decreasing and $m_\epsilon\to\infty$ if $\epsilon\to 0^+$.

A sequence\footnote{Note that we use here an a priori enumeration $\{\gamma_n\}_{n\in\bbN}$ of $\Gamma$.} $\{\chi_\gamma\}_{\gamma\in\Gamma}\subset\caH$ is \emph{complete} if
\[
\overline{{\rm span}\{\chi_\gamma\}_{\gamma\in\Gamma}} = \caH
\]
and it is \emph{minimal} if, for any $\gamma'\in\Gamma$,
\[
\chi_{\gamma'}\;\notin\;\overline{{\rm span}\{\chi_\gamma\}_{\gamma\in\Gamma\setminus\{\gamma'\}}}\;.
\]
Finally, a sequence $\{\chi_\gamma\}_{\gamma\in\Gamma}\subset\caH$ is a \emph{basis}  if there exists a unique sequence $\{\eta_\gamma\}_{\gamma\in\Gamma}\subset\caH$ such that
\begin{align}
\langle\eta_\gamma,\chi_{\gamma'}\rangle &=\delta_{\gamma,\gamma'}\;,\qquad \forall \gamma,\gamma'\in\Gamma\:, \label{eq:biortho}\\
\psi &=\sum_{\gamma\in\Gamma}\langle\eta_\gamma,\psi\rangle\chi_\gamma\;,\qquad \forall \psi\in\caH\;.\label{eq:basis}
\end{align}
We shall refer to~(\ref{eq:biortho}) as the biorthogonality condition. The convergence in \eqref{eq:basis} is meant as
\[
\lim_{N\to\infty}\Big\|\psi-\sum_{n=1}^N\langle\eta_{\gamma_n},\psi\rangle\chi_{\gamma_n}\Big\| = 0\;.
\]
A basis is complete and minimal. The uniqueness implies that if 
$\psi=\sum_{\gamma\in\Gamma}c_\gamma\chi_\gamma$ then necessarily 
$c_\gamma=\langle\eta_\gamma,\psi\rangle$.


A basis $\{\epsilon_\gamma\}_{\gamma\in\Gamma}\subset\caH$ is \emph{orthonormal} if it coincides with its biorthogonal system, \ie if 
$\langle\epsilon_\gamma,\epsilon_{\gamma'}\rangle=\delta_{\gamma,\gamma'}$. A basis $\{\chi_\gamma\}_{\gamma\in\Gamma}\subset\caH$ is a \emph{Riesz basis} if there is an orthonormal basis $\{\epsilon_\gamma\}_{\gamma\in\Gamma}\subset\caH$ and an invertible operator $W\in\bbG\bbL(\caH)$ such that $\chi_\gamma=W\epsilon_\gamma$. In such a case the biorthogonal system is given by 
\[
\eta_\gamma := (W^{-1})^*\epsilon_\gamma = |W^{-1}|^2\chi_\gamma\;,
\]
or for short by $\eta_\gamma=S^{-1}\chi_\gamma$ with $S:=WW^*$. In this case, the series \eqref{eq:basis} converges unconditionally (\ie converges for all permutations of the indices) for all $\psi\in\caH$ \cite[Theorem 3.6.3]{christensen-03} and (\ref{eq:basis}) makes sense without referring to an a priori enumeration of $\Gamma$.

\subsection{Frames and Gabor frames}
A sequence $\{\chi_\gamma\}_{\gamma\in\Gamma}\subset\caH$ is a \emph{frame}  if there exist finite constants $0< A\leqslant B$ such that
\begin{equation}\label{eq:frame}
A\|\psi\|^2\;\leqslant\;\sum_{\gamma\in\Gamma}|\langle\chi_\gamma,\psi\rangle|^2\;\leqslant\;B\|\psi\|^2,\qquad \forall \psi\in\caH\;.
\end{equation}
The numbers $A$ and  $B$ (which are not unique) are called \emph{frame bounds}.
By definition, it follows that a frame is complete. 

Let $\{\chi_\gamma\}_{\gamma\in\Gamma}\subset\caH$ be a frame and let $S:\caH\to\caH$ be the operator given by
\begin{equation}\label{eq:fr_op}
S\psi := \sum_{\gamma\in\Gamma}\langle\chi_\gamma,\psi\rangle\;\chi_\gamma\;.
\end{equation}
The operator $S$ is called the \emph{frame operator}. It is bounded, invertible, self-adjoint, and positive. In fact
\begin{equation*}
A{\bf 1}\leqslant S\leqslant\;B{\bf1},
\end{equation*}
see~\cite[Lemma 5.1.5]{christensen-03}. Moreover, the series $\sum_{\gamma\in\Gamma}\langle S^{-1}\chi_\gamma,\psi\rangle\chi_\gamma$ converges unconditionally for all $\psi\in\caH$ and 
\begin{equation}\label{eq:frame_basis}
\psi = \sum_{\gamma\in\Gamma}\langle S^{-1}\chi_\gamma,\psi\rangle\chi_\gamma\;,
\end{equation}
see~\cite[Theorem 5.1.6]{christensen-03}. The number $s_\gamma(\psi):=\langle S^{-1}\chi_\gamma,\psi\rangle$ are called \emph{frame coefficients} and $\{s_\gamma(\psi)\}_{\gamma\in\Gamma}\in\ell^2(\Gamma)$.

A Riesz basis is a frame \cite[Theorem 5.4.1]{christensen-03}. A frame that is not a Riesz basis is said to be an \emph{overcomplete frame}. In the latter case, there exist coefficients $\{z_\gamma\}_{\gamma\in\Gamma}\in\ell^2(\Gamma)\setminus\{0\}$ such that $\sum_{\gamma\in\Gamma}z_\gamma\chi_\gamma=0$. In particular, in the case of an overcomplete frame, the coefficients 
$\{c_\gamma\}_{\gamma\in\Gamma}\in\ell^2(\Gamma)$ are not uniquely determined by the representation $\psi=\sum_{\gamma\in\Gamma}c_\gamma\chi_\gamma$. However,  the
frame coefficients have the minimal $\ell^2$-norm among all sequences representing $\psi$, \cite[Lemma 5.4.2]{christensen-03} \ie
\[
\sum_{\gamma\in\Gamma}|\langle S^{-1}\chi_\gamma,\psi\rangle|^2\;\leqslant\;
\sum_{\gamma\in\Gamma}|c_\gamma|^2\;.
\]

Finally, \emph{Gabor frames} are special  frames in the Hilbert space $L^2(\bbR^\nu)$, see~\cite[Section 8.2]{christensen-03}. Let $\Gamma\subset \bbR^\nu$ a discrete and countable subset and $g\in L^2(\bbR^\nu)$ a fixed function. Then a frame $\{\chi_\gamma\}_{\gamma\in\Gamma}\subset L^2(\bbR^\nu)$ with elements given by
\[
\chi_\gamma(x) :=  {\rm e}^{{\rm i 2 \pi \sigma(\gamma, x)}}g(x-\gamma)\;,
\]  
with $\sigma:\bbR^\nu\times\bbR^\nu\to\bbR$ a bilinear and anti-symmetric map, is called   the (irregular) Gabor frame defined by $\Gamma$ and $g$. When $\Gamma$ is a lattice in $\bbR^\nu$ then the associated Gabor frame is called \emph{regular}.

\subsection{Example: The magnetic Gabor frame}\label{sec:MagGabor}
 Let us introduce a recurring example that will be relevant for the applications of the quantum dynamics in a constant magnetic field discussed in detail in Section \ref{sec: Laudau Hamiltonian}.  
Consider the closed subspace $\caH_0\subset L^2(\bbR^2)$ defined by
\begin{equation}\label{eq:descr_H_0}
\caH_0=\Big\{f(x){\rm e}^{-\omega^2 |x|^2}
\;\Big|\; f(x) = \sum_{m\in\bbN_0}a_m(x_1-{\rm i} x_2)^m \Big\}\;
\end{equation}
where $\bbN_0:=\bbN\cup\{0\}$ and $x=(x_1,x_2)\in\bbR^2$. In physical applications,  the parameter $\omega>0$ is related to the concept of \emph{magnetic length}  through the relation $\ell_B:=(2\omega)^{-1}$. We will refer to $\caH_0$ as the \emph{first Landau level} for reasons that will be clarified in Section \ref{sec: Laudau Hamiltonian}.
Let $\omega>0$ be fixed and let $g_{\omega}$ be the normalized function
\[
g_{\omega}(x) := \omega\sqrt{\frac{2}{\pi}}{\rm e}^{-\omega^2 |x|^2}\;.
\]
For any $\alpha,\beta>0$, let $\Gamma_{(\alpha,\beta)}:=\alpha \bbZ\times \beta \bbZ\subset\bbR^2$. For all $\gamma\in \Gamma_{(\alpha,\beta)}$, we define
\[
\chi_\gamma(x) := {\rm e}^{-{\rm i}2\omega^2\left(\gamma\wedge x\right)}g_{\omega}
\left(x-\gamma\right)\;,\qquad 
\]
where $(\gamma_1,\gamma_2)\wedge(x_1,x_2):=\gamma_1x_2-\gamma_2x_1$.
These functions are not orthogonal but are asymptotically orthogonal in the sense that
\begin{equation}\label{eq:gaus_orth}
\langle\chi_\gamma,\chi_{\gamma'}\rangle = {\rm e}^{-\omega^2 |\gamma-\gamma'|^2}\;,
\end{equation}
see~Lemma~\ref{prop:01}. We will show in Section \ref{sec: Laudau Hamiltonian} that the sequence $\{\chi_\gamma\}_{\gamma\in\Gamma_{(\alpha,\beta)}}$ may or may not be a frame for $\caH_0$ depending on the value of the product $\alpha\beta$. There are three distinct cases:
\begin{itemize}
\item[(1)] If $\alpha\beta<\frac{\pi}{2\omega^2}$ the family $\{\chi_\gamma\}_{\gamma\in\Gamma_{(\alpha,\beta)}}$ is overcomplete in the sense that any vector of the family is contained in the closed linear span of the other vectors. In particular, it is a regular Gabor frame
and this property is preserved if one removes from the family any finite number of vectors.

\item[(2)] If $\alpha\beta=\frac{\pi}{2\omega^2}$ the family $\{\chi_\gamma\}_{\gamma\in\Gamma_{(\alpha,\beta)}}$ is complete but it is not a frame. It remains complete if one single vector is removed but becomes incomplete if any two vectors are removed.

\item[(3)] If $\alpha\beta>\frac{\pi}{2\omega^2}$ the family $\{\chi_\gamma\}_{\gamma\in\Gamma_{(\alpha,\beta)}}$ is incomplete.
\end{itemize}
We will be mainly interested in case (1) and we shall refer to the family $\{\chi_\gamma\}_{\gamma\in\Gamma_{(\alpha,\beta)}}$ as the \emph{magnetic Gabor frame}. The threshold case (2) is known in the literature as a von Neumann lattice. In this case, by removing from the family 
$\{\chi_\gamma\}_{\gamma\in\Gamma_{(\alpha,\beta)}}$ a single vector one obtains a minimal and complete system. The discussion of this case, although interesting, is outside the scope of this work. Finally, case (3) will not be interesting for us in view of the lack of completeness.

\subsection{\bf Localized  frames} 
Let $\{\chi_\gamma\}_{\gamma\in\Gamma}\subset\s{H}$ be a {frame} 
  labelled by the 
discrete and countable set $\Gamma$  with a distance $d:\Gamma\times \Gamma\to[0,+\infty)$ satisfying
property \eqref{eq:bound_cond}.
\begin{defn}[Localized frame]\label{def:loc_fr}
A frame $\{\chi_\gamma\}_{\gamma\in\Gamma}\subset\s{H}$ is called (exponentially) localized if 
\begin{equation}\label{eq:eq:ineq_asum}
|\langle\chi_\gamma,\chi_{\gamma'}\rangle|\;\leqslant\;G \expo{-\lambda d(\gamma,\gamma')}\;,\qquad\forall\; \gamma,\gamma'\in\Gamma
\end{equation}
for some $G\geqslant 1$ and $\lambda>0$. For convenience we will assume the normalization of the frame \ie $\|\chi_\gamma\|=1$ for every $\gamma\in\Gamma$.
\end{defn}

\begin{ex}\label{Rk_exp_gab}
The magnetic Gabor frame described in Section \ref{sec:MagGabor} provides an example of a localized frame. More generally,
let $\Gamma=\alpha_1\Z\times\ldots\times\alpha_\nu\Z$, with $\alpha_j>0$ for all $j=1,\ldots,\nu$, and let $\alpha_\ast:=\min\{\alpha_1,\ldots,\alpha_\nu\}$. Assume that
\begin{equation}\label{eq:mic_mol}
|\langle\chi_\gamma,\chi_{\gamma'}\rangle|\;\leqslant\;\expo{-\varpi |\gamma-\gamma'|^2}\;,\qquad \forall\;\gamma,\gamma'\in\Gamma\;
\end{equation}
for a given $\varpi>0$. By observing that for every $\gamma=(\alpha_1 n_1,\ldots,\alpha_\nu n_\nu)$
\[
|\gamma|^2 = \sum_{j=1}^\nu\alpha_j^2|n_j|^2\;\geqslant\;\sum_{j=1}^\nu\alpha_j^2|n_j|\;\geqslant\;\alpha_\ast\sum_{j=1}^\nu| \alpha_j n_j| = \alpha_\ast |\gamma|_1
\]
and introducing the distance $d(\gamma,\gamma') := \alpha_\ast |\gamma-\gamma'|_1$, we conclude from \eqref{eq:mic_mol} that
\begin{equation*}
|\langle\chi_\gamma,\chi_{\gamma'}\rangle|\;\leqslant\;\expo{-\lambda d(\gamma,\gamma')}\;,\qquad \forall\;\gamma,\gamma'\in\Gamma\;
\end{equation*}
with $\lambda:=\varpi\alpha_\ast$. Since the distance is invariant under translation one gets
\[
\begin{aligned}
\sum_{\xi\in\Gamma}\expo{-\epsilon d(\gamma,\xi)} 
&= \sum_{\xi\in\Gamma}\expo{-\epsilon d(0,\xi)} = \prod_{j=1}^\nu\sum_{n\in\Z}\Big(\expo{-\epsilon\alpha_\ast\alpha_j}\Big)^{|n|}
\leqslant2^\nu\prod_{j=1}^\nu\sum_{j=0}^{+\infty}\Big(\expo{-\epsilon\alpha_\ast\alpha_j}\Big)^{n}\\
&\leqslant2^\nu\Big(\sum_{j=0}^{+\infty}\Big(\expo{-\epsilon\alpha_\ast^2}\Big)^{n}\Big)^\nu = 
\Big(\frac{2}{1-\expo{-\epsilon\alpha_\ast^2}}\Big)^\nu,
\end{aligned}
\]
which yields an explicit upper bound on $m_\epsilon$ in~\eqref{eq:bound_cond}. \hfill $\blacktriangleleft$
\end{ex}

\begin{ex}\label{Rk_exp_gab_2}
For the aims of this work, it is relevant to consider a slight modification of the above frame. Let us consider the discrete set $\widetilde{\Gamma}:=\N_0\times\Gamma$ and we denote $\widetilde{\gamma}:=(r,\gamma)$ a generic point of $\widetilde{\Gamma}$. The extended frame $\{\chi_{\widetilde{\gamma}}\}_{\widetilde{\gamma}\in\widetilde{\Gamma}}$ satisfies
\begin{equation*}
|\langle\chi_{\widetilde{\gamma}},\chi_{\widetilde{\gamma}'}\rangle|\;\leqslant\;\delta_{r,r'}\expo{-\varpi |\gamma-\gamma'|^2}\;,\qquad \forall\;\widetilde{\gamma},\widetilde{\gamma}'\in\widetilde{\Gamma}\;
\end{equation*}
for a given $\varpi>0$. As in the above example, we obtain
\[
|\langle\chi_{\widetilde{\gamma}},\chi_{\widetilde{\gamma}'}\rangle|\;\leqslant\;
\expo{-\lambda \widetilde{d}(\widetilde{\gamma},\widetilde{\gamma}')}\;,\qquad \forall\;\widetilde{\gamma},\widetilde{\gamma}'\in\widetilde{\Gamma}\;,
\]
where $\widetilde{d}(\widetilde{\gamma},\widetilde{\gamma}') := |r-r'|+d(\gamma,\gamma')$, namely $\{\chi_{\widetilde{\gamma}}\}_{\widetilde{\gamma}\in\widetilde{\Gamma}}$ is again a localized frame, for the same constant $m_\epsilon$. The definition of the vectors $\widetilde{\gamma}$ and the physical relevance of this example will be clarified in Section~\ref{sec: Laudau Hamiltonian}, where we will also show that $L^2(\R^2)=\bigoplus_{r\in\N_0}\s{H}_{r}$, where
\[
\caH_ r := \overline{{\rm span}\{\chi_{(r,\gamma)}\}_{\gamma\in\Gamma}}
\]
is called the $r$-th \emph{Landau level}.
 \hfill $\blacktriangleleft$
\end{ex}

The representation~(\ref{eq:frame_basis}) underlines the importance of the operator $S^{-1}$. It turns out that the localization of the frame carries over to the matrix elements of $S^{-1}$ and its positive powers.
\begin{prop}\label{Prop: localization of S^-p}
Let $\{\chi_\gamma\}_{\gamma\in\Gamma}\subset\s{H}$ be a localized frame with localization rate $\lambda>0$, see Definition~\ref{def:loc_fr}, and let $S$ be the associated frame operator. Then, for every $p\in\N$ there exist constants $0<\lambda_p<\lambda$ and $a_p>0$ such that
\begin{equation}\label{matrix elements of S-p}
\left|\langle\chi_\gamma,S^{-p}\chi_{\gamma'}\rangle\right|\;\leqslant\; a_p \expo{-\lambda_p d(\gamma,\gamma')}
\end{equation}
for all $\gamma,\gamma'\in\Gamma$, where $S^{-p}:=(S^{-1})^p$.
\end{prop}

The proof of this result is quite technical and we postpone it until Section~\ref{sec:add_proof}.
The result above has an immediate consequence that will be used in the following part of this work. 

\begin{cor}\label{Cor_1}
Under the hypotheses of Proposition~\ref{Prop: localization of S^-p}, let $H = H^*$ be a self-adjoint operator with dense domain $\s{D}(H)\subset\s{H}$ such that:
\begin{enumerate}
\item $H\chi_\gamma=E(\gamma)\chi_\gamma$ with $E(\gamma)\in\R$ 
for every $\gamma\in\Gamma$; 
\item $S\s{D}(H)\subseteq \s{D}(H)$.
\end{enumerate}
Then,
\[
\left|\langle\chi_\gamma,S^{-1}HS^{-1}\chi_{\gamma'}\rangle\right|\;\leqslant\; \delta_{E(\gamma),E(\gamma')} \; a_2  E(\gamma) \expo{-\lambda_2 d(\gamma,\gamma')}\;,\qquad \;.
\]
for all $\gamma,\gamma'\in\Gamma$.
\end{cor}
\begin{proof}
If $\psi\in \s{D}(H)$, then by~\eqref{eq:fr_op}, the self-adjointness of $H$ and (i,ii),
\begin{equation*}
SH\psi = \sum_{\gamma\in\Gamma}\langle \chi_\gamma, H\psi\rangle \chi_\gamma = \sum_{\gamma\in\Gamma} E_\gamma \langle \chi_\gamma, \psi\rangle \chi_\gamma
= \sum_{\gamma\in\Gamma}  \langle \chi_\gamma, \psi\rangle H\chi_\gamma = HS\psi\:,
\end{equation*}
namely $SH=HS$ on $\caD(H)$. By functional calculus this implies that $S^{-1}H=HS^{-1}$ and so
\[
\langle\chi_\gamma,S^{-1}HS^{-1}\chi_{\gamma'}\rangle = \delta_{E(\gamma),E(\gamma')} \; E(\gamma) \langle\chi_\gamma,S^{-2}\chi_{\gamma'}\rangle\;.
\]
The result follows from Proposition~\ref{Prop: localization of S^-p}.
\end{proof}

Let $\Gamma_E:=\{\gamma\in\Gamma\;|\; E(\gamma)=E\}$ with $E\in\sigma(H)$. Then $\Gamma = \bigcup_{E\in \sigma(H)}\Gamma_E$ and Corollary \ref{Cor_1} states that inside any $\Gamma_E$ the decay of the matrix elements $\langle\chi_\gamma,S^{-1}HS^{-1}\chi_{\gamma'}\rangle$ is exponential, and between different $\Gamma_E$'s the matrix elements vanish.

\addtocounter{ex}{-1}
\begin{ex}[continued]
Corollary \ref{Cor_1} is tailored to the decomposition of $L^2(\bbR^2)$, where each $\caH_r$ is the eigenspace of the \emph{Landau Hamiltonian}, namely $E(\widetilde \gamma) = q(r)$, where $q:\N\to \R$ is an increasing function. The partition of the label set is in terms of the Landau level $r$: $\widetilde \Gamma = \cup_{r\in\bbN}\left(\{r\}\times\Gamma\right)$.  
 \hfill $\blacktriangleleft$
\end{ex}

\section{A Lieb-Robinson bound for localized frames}   \label{sec: LRB}

We are ready to prove a Lieb-Robinson bound in the continuum. For this, we require additional mild regularity on the metric space $(\Gamma,d)$, namely that it is $\nu$-dimensional in the following sense: There is $\nu\in\bbN$ and $\kappa>0$ such that 
\begin{equation}\label{nu-dimensional}
\sup_{\gamma\in\Gamma}\vert \{\gamma'\in\Gamma: d(\gamma',\gamma)\leq r\}\vert\leqslant \kappa r^\nu
\end{equation}
for all $r\geqslant 0$.

\subsection{Interacting Hamiltonians for many-body fermions}   \label{sub: Hamiltonian}

\subsubsection{The CAR algebra}

Let $\caH$ be a separable Hilbert space. The CAR algebra $\caA(\caH)$ is the unital C*-algebra generated by $\bf1$ and
\begin{equation*}
\{a(\psi),a\str(\psi)\:\vert\:\psi\in\caH\}
\end{equation*}
subject to the relations
\begin{align}\label{CAR}
& a(\psi) a\str(\phi) + a\str(\phi)a(\psi) = \langle \psi,\phi\rangle \bf 1,\\
& a(\psi) a(\phi) + a(\phi)a(\psi) = 0 = a\str(\psi) a\str(\phi) + a\str(\phi)a\str(\psi),
\end{align}
which are usually called the \emph{canonical anticommutation relations}. Here $\psi\mapsto a(\psi)$ is antilinear and~(\ref{CAR}) immediately implies that $\Vert a(\psi)\Vert = \Vert\psi\Vert$, where the norm on the left is the C*-norm, while that on the right is the Hilbert space norm. 

\subsubsection{Free dynamics} 

Let $t\mapsto U_t$ be a strongly continuous one-parameter group of unitaries in $\caH$. Then 
\begin{equation}\label{Free D}
t\mapsto \tau_t(a(\phi)) := a(U_t\phi)\:,
\end{equation}
with $\tau_t(a\str(\phi)):= (\tau_t(a(\phi)))\str$, defines a strongly continuous family of *-automorphisms, namely a \emph{dynamics}, of $\caA(\caH)$. One defines the densely defined derivation $\delta$ on $\caA(\caH)$ by
\begin{equation*}
\frac{d}{dt}\tau_t(A)\vert_{t=0} = \delta(A).
\end{equation*}

If $H = H\str$ is the (in general unbounded) generator of $U_t$, we obtain a formal expression for~$\delta$:
\begin{equation*}
\delta(a(\phi)) = a(\iu H\phi).
\end{equation*}
Continuing along such formal manipulations, we let
\begin{equation}\label{Second Q on Basis}
\bbH = \sum_{n,m}h_{n,m}a\str(\psi_n)a(\psi_m),\qquad h_{n,m} = \langle \psi_n, H\psi_m\rangle,
\end{equation}
where $\{\psi_n\}_{n\in\bbN}$ is an orthnormal basis of $\caH$. We then note that $\bbH = \sum_m a\str(H\psi_m)a(\psi_m)$, and the anticommutation relations together with the identity $[A_1A_2,A_0] = A_1\{A_2,A_0\} - \{A_1,A_0\}A_2$ then yield
\begin{align*}
\iu[\bbH,a(\phi)] &= \iu\sum_m\left(a\str(H\psi_m)\{a(\psi_m),a(\phi)\} - \{ a\str(H\psi_m),a(\phi)\}a(\psi_m)\right) \\
&=a\Big(\iu\sum_m \langle H\psi_m,\phi\rangle\psi_m\Big) = a(\iu H \phi)
\end{align*}
since $H$ is self-adjoint. Hence,
\begin{equation*}
\delta(A) = \iu[\bbH,A].
\end{equation*}
The map $H\mapsto\bbH$ is usually referred to as the \emph{second quantization}. It is rather obvious that~(\ref{Second Q on Basis}) cannot in general be expected to be a convergent series in the C*-norm. The formal expression may however make sense as a densely defined self-adjoint operator in a particular representation of~$\caA(\caH)$, see~Section~\ref{sec: Fock quantization}.

Our first goal is much more general: We shall define a general dynamics of~$\caA(\caH)$ generated by formal expressions as~(\ref{Second Q on Basis}). While the quadratic expression reflects the absence of interactions, our methods allow for a very general class of truly interacting systems. This is well understood in the context of lattice systems --- spin systems or lattice fermions. A localized frame $\{\chi_\gamma\}_{\gamma\in\Gamma}$ will allow us to lift these methods to fermion systems in the continuum.

The following result follows from \cite[Theorem 5.2.5]{BratRob2} along with the identity~(\ref{eq:frame_basis}).
\begin{lemma}\label{lem:density}
Let $\mathfrak{U}_\Gamma\subset\caA(\caH)$ be the subalgebra of polynomials in $\{a(\chi_\gamma),a^*(\chi_\gamma):\gamma\in\Gamma\}$. Then $\mathfrak{U}_\Gamma$ is dense in $\caA(\caH)$.
\end{lemma}

Concretely, for every $\psi_m$ of the basis, (\ref{eq:frame_basis}) yields that
\begin{equation}\label{Expansion in frame annihilation}
a(\psi_m) = \sum_{\gamma\in\Gamma} \overline{s_\gamma(m)}  a(\chi_\gamma)
\end{equation}
where $s_\gamma(m) = \langle S^{-1}\chi_\gamma,\psi_m\rangle$, and so (\ref{Second Q on Basis}) can equivalently be written (formally) as
\begin{equation}\label{Second Q on Frame}
\bbH = \sum_{\gamma,\gamma'\in\Gamma}t_H(\gamma',\gamma)a\str(\chi_{\gamma'})a(\chi_\gamma)
\end{equation}
where
\begin{equation*}
t_H(\gamma',\gamma)\;:=\;\sum_{n,m\in\N}h_{n,m}{s_{\gamma'}(n)}\overline{s_\gamma(m)},
\end{equation*}
for every $\gamma,\gamma'\in\Gamma$. 

\begin{prop}
Let $H$ be a self-adjoint operator and assume that $S^{-1}\chi_\gamma\in\s{D}(H)$ for every $\gamma\in\Gamma$.
Then,
\begin{equation*}
t_H(\gamma',\gamma)\;=\;\left\langle \chi_{\gamma'},V_H\chi_\gamma\right\rangle
\end{equation*}
where $V_H:=(S^{-1})HS^{-1}$. If the assumptions of Corollary~\ref{Cor_1} hold, then
\[
|t_H(\gamma',\gamma)|\;\leqslant\;  \delta_{E(\gamma),E(\gamma')} \; a_2 E(\gamma) \expo{-\lambda_2 d(\gamma,\gamma')}\;,\qquad \gamma,\gamma'\in\Gamma\;.
\]
\end{prop}

\begin{proof}
For any $N,M\in\bbN$, 
\begin{equation*}
t_{H}^{(N,M)}(\gamma',\gamma) := \sum_{n=1}^N\sum_{m=1}^M h_{n,m}{s_{\gamma'}(n)}\overline{s_\gamma(m)}
=\left\langle S^{-1}\chi_{\gamma'}, P_NHP_MS^{-1}\chi_\gamma\right\rangle\\
\end{equation*}
where $P_N:=\sum_{n=1}^N\ketbra{\psi_n}{\psi_n}$ is the projection onto the span of the first $N$ elements of the basis. The results follow by observing that $P_N$ converges strongly to the identity.
\end{proof}

In other words: If a localized frame consists of eigenvectors of a self-adjoint Hamiltonian, then the coefficients in a formal second quantization  which uses the creation and annihilation operators defined on the frame~(\ref{Second Q on Frame}) decay exponentially. This observation motivates the assumptions of the next section. We now first prove a Lieb-Robinson bound and then use it to define the dynamics at the level of the algebra. Note that the latter fact, namely the existence of the dynamics in the infinite volume limit, is trivial in the case just discussed of a free dynamics~(\ref{Free D}); It is in general an open question for fermionic systems in the continuum.

\subsubsection{Interaction}

In order to make sense of a Hamiltonian of the form~(\ref{Second Q on Frame}), we use the concept of an interaction, namely of a map $\Phi$ that attaches to any finite subset $Z\subset\Gamma$ a self-adjoint elements in $\caA_Z$. Here, $\caA_Z$ denotes the set of elements of $\caA(\caH)$ that can be expressed as a polynomial in $\{a_\gamma,a\str_\gamma\:\vert\:\gamma\in Z\}$. Here and from now on, we use the shorthand $a_\gamma := a(\chi_\gamma)$. Specifically, we consider an even interaction $\Phi: \Gamma \to \caA$ of the form
\begin{equation}\label{interaction}
\Phi(Z) = \sum_{k=1}^{|Z|} f_k(Z)(M_k(Z) + M_k(Z)^*)
\end{equation}
where  $M_k(Z)\in\caA_Z$ is a monic monomial of degree $2k$. Our proofs below would hold for slightly more general interactions than~\eqref{interaction} in the sense that we could handle more than just one term for each degree $k$, but this covers the physically relevant cases. Note that for any $Z$ the sum over $k$ is necessarily finite by the CAR relations: For each site one can accommodate a monomial of degree at most $2$ and in turn the maximum degree of a monomial over $Z$ is $2|Z|$.

While~(\ref{interaction}) determines the general form of the interactions, we need to impose conditions of sufficient decay of $\Vert \Phi(Z)\Vert$ with the size of the set $Z$. We express this as a condition on the function $f$ since $\Vert M_k(Z)\Vert = 1$. 
\begin{assumption}\label{Hyp:Interaction}
Let $\Phi$ be an interaction of the form~(\ref{interaction}). There is $\xi_0>0$ such that for all $0<\zeta<\xi<\xi_0$, 

\begin{equation}\label{Assumption Interaction}
C_\Phi(\zeta,\xi):=\sup_{\gamma\in\Gamma}\sup_{Z'\subset\Gamma}\frac{\ep{\zeta d(\gamma,Z')}}{\caD(Z')}\sum_{Z\subset \Gamma}\sum_{k=1}^{|Z|} k^2 f_k(Z)\caD(Z)\ep{-\zeta d(\gamma,Z)}\ep{-\xi d(Z,Z')}  <\infty
\end{equation}
where $\caD(Z) = (1+\mathrm{diam}(Z))^\nu$ and $\nu\in\bbN$ is the same of \eqref{nu-dimensional}.
\end{assumption}
\noindent By picking $Z' = \{\gamma'\}$, we see that this condition implies in particular that
\begin{equation}\label{Assumption Interaction Point}
\sum_{Z\subset \Gamma}\sum_{k=1}^{|Z|} k^2 f_k(Z)\caD(Z)\ep{-\zeta d(\gamma,Z)} \ep{-\xi d(Z,\gamma')}\leqslant C_\Phi(\zeta,\xi) \ep{-\zeta d(\gamma,\gamma')}
\end{equation}
for all $\gamma,\gamma'\in\Gamma$. 

Let us consider the simple case of a density-density interaction
\begin{equation*}
\sum_{\gamma,\gamma'\in\Gamma}f(\gamma,\gamma')n(\gamma) n(\gamma'),
\end{equation*}
where $n(\gamma) = a\str_\gamma a_\gamma$, namely
\begin{equation}\label{2B interaction}
f_k(Z) = \begin{cases}
f(x,y) &\text{if }Z =\{x,y\}, k=2 \\ 0 & \text{otherwise}.
\end{cases}
\end{equation}
In this case, the parameter $\xi$ of Assumption~\ref{Hyp:Interaction} determines the minimal rate of exponential decay of the interaction. 
\begin{lemma}
Assume that the interaction $\Phi$ is such that~(\ref{2B interaction}) holds and that there is $\mu>0$ such that 
\begin{equation}\label{2Bd Decay}
f(x,y) \leqslant f_0\ep{-\mu d(x,y)},\qquad \forall x,y\in\Gamma.
\end{equation}
Then Assumption~\ref{Hyp:Interaction} holds for all $\xi_0\leqslant\mu $.
\end{lemma}
\begin{proof}
We start with a geometric inequality. Let $x,y,\gamma\in\Gamma$ and $Z\subset\Gamma$, and we assume w.l.o.g\ that $d(\gamma,\{x,y\}) = d(\gamma,x)$. If $x_0\in Z$ is such that $d(x,Z) = d(x,x_0)$. Then 
\begin{equation*}
d(\gamma,Z)\leqslant d(\gamma, x_0)\leqslant d(\gamma,x) + d(x,x_0) = d(\gamma,x) + d(x,Z).
\end{equation*}
Since, moreover, $d(x,Z)\leqslant d(x,y) + d(\{x,y\},Z)$, we conclude that 
\begin{equation*}
d(\gamma,Z)\leqslant d(\gamma,\{x,y\}) + d(x,y) + d(\{x,y\},Z).
\end{equation*}
This allows to conclude that
\begin{multline*}
\ep{\zeta d(\gamma,Z')}\sum_{\{x,y\}\subset\Gamma}f(x,y)(1+d(x,y))^\nu\ep{-\zeta d(\gamma,\{x,y\})}\ep{-\xi d(\{x,y\},Z')} \\
\leq \sum_{\{x,y\}\subset\Gamma}f(x,y)(1+d(x,y))^\nu\ep{\zeta d(x,y)}\ep{-(\xi-\zeta) d(\{x,y\},Z')}.
\end{multline*}
It remains to bound the last sum by a constant times $\caD(Z')$. The smallest ball $B(Z')$ containing the set $Z'$ has radius $\mathrm{diam}(Z')/2$. Since $Z'$ is fixed, we decompose the sum in terms of the point of $\{x,y\}$ that is closest to $Z'$, which we call $x$, namely whether $x\in B(Z')$ or $x\in \Gamma\setminus B(Z')$. For $x\in B(Z')$, we use the trivial bound $\ep{-(\xi-\zeta) d(\{x,y\},Z')}\leq 1$ to obtain
\begin{equation*}
\sum_{x\in B(Z')}\sum_{y\in\Gamma}[\cdots]
\leq f_0\vert B(Z')\vert \sup_{x\in\Gamma}\sum_{y\in\Gamma}(1+d(x,y))^\nu\ep{-(\mu-\zeta) d(x,y)},
\end{equation*}
which is of the right form since~(\ref{nu-dimensional}) relates $\vert B(Z')\vert$ with $\caD(Z')$, and the fact that the series is convergent for every $\zeta<\mu$ and the supremum is bounded in view of \eqref{eq:bound_cond}.  The rest of the sum is estimated by decomposing it into concentric shells of radii $\{(s+\mathrm{diam}(Z')/2):s\in\bbN\}$. On each shell, $d(\{x,y\},Z') = d(x,Z')$ implies that $\ep{-(\xi-\zeta) d(\{x,y\},Z')}\leqslant \ep{-(\xi-\zeta) s}$ and the volume of the shell is proportional to $(s+\mathrm{diam}(Z')/2)^{\nu-1}$, so that
\begin{equation*}
\sum_{x\in \Gamma\setminus B(Z')}\sum_{y\in\Gamma}[\cdots]
\leqslant C \sum_{s=1}^\infty (s+\mathrm{diam}(Z')/2)^{\nu-1} \ep{-(\xi-\zeta) s}f_0\sup_{x\in\Gamma}\sum_{y\in\Gamma}(1+d(x,y))^\nu\ep{-(\mu-\zeta) d(x,y)}.
\end{equation*}
The second sum is finite as above. The first one decomposes into $\nu - 1$ finite contributions, each of them being bounded by $\caD(Z')\sum_{s}s^{\nu-1}\ep{-\zeta s}<\infty$. This concludes the proof. 
\end{proof}

\subsection{The Lieb-Robinson bound}   \label{sub: LR proof}

The Lieb-Robinson bound is best understood as a statement on the almost preservation of the algebraic structure (here the anticommutation relations) by the dynamics. As usual, we state the result first as a finite-volume bound that is however uniform in the volume.

Given an interaction $\Phi$ and a finite subset $\Lambda\subset\Gamma$, the Hamiltonian $H_\Lambda$ is the self-adjoint operator
\begin{equation*}
H_\Lambda := \sum_{Z\subset\Lambda}\Phi(Z)
\end{equation*}
and the finite volume dynamics is the family of automorphisms of $A\in\caA(\caH)$ defined by $\tau^\Lambda_t(A) = \ep{\iu t H_\Lambda}A\ep{-\iu t H_\Lambda}$.
\begin{thm}\label{thm: LR bound}
Let $\xi>0$. Let $\{\chi_\gamma\}_{\gamma\in\Gamma}$ be a localized frame~(\ref{eq:eq:ineq_asum}) with $\lambda = \xi$. Let $\Phi$ be an interaction of the form~(\ref{interaction}) such that Assumption~\ref{Hyp:Interaction} holds. Then for any $\zeta<\xi$
\begin{equation*}
\left\Vert\left\{\tau_t^\Lambda\left(a\str_{\gamma}\right),a_{\gamma'}\right\}\right\Vert
\leqslant G \ep{-\zeta\left(d(\gamma,\gamma') - v | t |\right)}
\end{equation*}
for any $\gamma,\gamma'\in\Gamma$, where $v = \frac{16 G C_\Phi(\zeta,\xi)}{\zeta}$.
The bound also holds with any other combination of creation and annihilation operators. 
\end{thm}
\begin{proof}
For any $\gamma,\gamma'\in\Gamma$, we let
\begin{equation*}
F_t^{\Lambda,\sharp}(\gamma,\gamma') := \sup_{\vert\alpha\vert + \vert\beta\vert\leq 1}
\left\Vert\left\{\tau_t^\Lambda\left(\alpha a\str_{\gamma} + \overline \beta a_{\gamma}\right),a^\sharp_{\gamma'}\right\}\right\Vert
\end{equation*}
where $a^\sharp_{\gamma'}$ stands for either $a_{\gamma'}$ or $a\str_{\gamma'}$, and immediately note that
\begin{equation*}
0\leqslant F_0^{\Lambda,\sharp}(\gamma,\gamma') \leqslant \vert G(\gamma,\gamma')\vert ,\qquad G(\gamma,\gamma') : = \langle \chi_\gamma,\chi_\gamma'\rangle.
\end{equation*}
We expand the second term of the Duhamel identity
\begin{equation}\label{LR Starting Point}
\left\{\tau_t^\Lambda (a^\sharp_{\gamma}),a^\sharp_{\gamma'}\right\}
= \left\{a^\sharp_{\gamma},a^\sharp_{\gamma'}\right\} + \sum_{Z\subset\Lambda}\iu\int_0^t \left\{\tau_s^\Lambda ([\Phi(Z),a^\sharp_{\gamma}]),a^\sharp_{\gamma'}\right\}ds.
\end{equation}
Since the monomials $M_k(Z)$ are of degree $2k$, see~(\ref{interaction}), Leibniz' rule implies that
\begin{equation}\label{expand Mk}
\Big[M_k(Z),a^\sharp_{\gamma}\Big] = \sum_{j=1}^k(a^\sharp_{z_1^1}a^\sharp_{z_1^2})
\cdots (a^\sharp_{z_{j-1}^1}a^\sharp_{z_{j-1}^2})
\Big[(a^\sharp_{z_{j}^1}a^\sharp_{z_{j}^2}),a^\sharp_{\gamma}\Big] \\
(a^\sharp_{z_{j+1}^1}a^\sharp_{z_{j+1}^2})
\cdots (a^\sharp_{z_{k}^1}a^\sharp_{z_{k}^1})
\end{equation}
where $z^{1}_j,z^{2}_j\in Z$ for all $j\in\{1,\ldots,k\}$. Using that $[A_1A_2,A_0] = A_1\{A_2,A_0\} - \{A_1,A_0\}A_2$, the CAR relations yield that the commutator is a sum of at most $2k$ terms of the form
\begin{equation*}
\pm G(z,\gamma) N_k(Z)\text{ or }\pm \overline{G(z,\gamma)} N_k(Z)
\end{equation*}
where $N_k(Z)$ are monomials of degree $2k-1$ in $\caA_Z$. Each term $\{\tau_s^\Lambda ([M_k(Z)^\sharp,a^\sharp_{\gamma}]),a^\sharp_{\gamma'}\}$ of~(\ref{LR Starting Point}) can further be expanded using the identity
\begin{equation*}
\{A_1\cdots A_{2k-1},A_0\} = \sum_{j=1}^{2k-1}  (-1)^{j+1}A_1\cdots A_{j-1} \{A_j,A_0\} A_{j+1} \cdots A_{2k-1}
\end{equation*}
which is easily checked by induction. This yields in total at most $2k(2k-1)$ terms, the norm of which being bounded above by
\begin{equation*}
\vert G(z,\gamma)\vert \left\Vert\left\{\tau_s^\Lambda(a^\sharp_{z'} ),a^\sharp_{\gamma'}\right\}\right\Vert
\end{equation*}
where $z,z'\in Z$ are not necessarily distinct. Here, we used the fact that $\tau_s^\Lambda$ is an automorphism and that fermionic creation and annihilation operators are normalized. We conclude that
\begin{equation*}
F_t^{\Lambda,\sharp}(\gamma,\gamma')
\leqslant \vert G(\gamma,\gamma')\vert 
 + 4 \sum_{Z\subset\Lambda}\sum_{k=1}^\infty 2k(2k-1)f_k(Z) \sup_{z,z'\in Z}\vert G(\gamma,z)\vert \int_0^{\vert t\vert} F_s^{\Lambda,\sharp}(z',\gamma')ds.
\end{equation*}
The factor $4$ arises here from the fact that both $M_k(Z)$ and its adjoint must be taken into account, and that $F_t^{\Lambda,\sharp}(\gamma,\gamma')$ contains both a creation and an annihilation operator at $\gamma$. This is readily iterated to
\begin{equation}\label{LR Iteration}
F_t^{\Lambda,\sharp}(\gamma,\gamma') \leqslant \vert G(\gamma,\gamma')\vert + \sum_{j=1}^\infty \frac{a_j(\gamma,\gamma') \vert t\vert ^j}{j!}
\end{equation}
where
\begin{equation*}
a_1(\gamma,\gamma') = 4 \sum_{Z\subset\Lambda}\sum_{k=1}^{|Z|} 2k(2k-1)f_k(Z) \sup_{z,z'\in Z}\vert G(\gamma,z)\vert\vert G(z',\gamma') \vert
\end{equation*}
and
\begin{multline*}
a_j(\gamma,\gamma') = 4^j\sum_{Z_j\subset\Lambda}\sum_{k_j=1}^{|Z_j|}\sup_{z_j,z_j'\in Z_j}\cdots\sum_{Z_1\subset\Lambda}\sum_{k_1=1}^{|Z_1|} \sup_{z_1,z_1' \in Z_1}\\
 \vert G(\gamma,z_1) \vert \bigg( \prod_{i=2}^j \vert G(z_{i-1}',z_i)\vert  \bigg)  \vert G(z_{j}',\gamma') \vert 
  \prod_{i=1}^j 2k_i(2k_i-1)f_{k_i}(Z_i).
\end{multline*}
We now use that $\{\chi_\gamma\}_{\gamma\in\Gamma}$ is a localized frame~(\ref{eq:eq:ineq_asum}) to conclude that
\begin{equation*}
a_1(\gamma,\gamma') \leqslant 4 G^2 \sum_{Z\subset\Lambda}\sum_{k=1}^{|Z|} 2k(2k-1)f_k(Z) \ep{-\xi d(\gamma,Z)}\ep{-\xi d(Z,\gamma')}
\leqslant 16 G^2 C_\Phi(\zeta,\xi) \ep{-\zeta d(\gamma,\gamma')}
\end{equation*}
where we used that $\caD(Z)\geq 1$ and~(\ref{Assumption Interaction Point}). For the higher order terms, we proceed similarly, using first that the frame is localized to conclude that 
\begin{equation*}
a_j(\gamma,\gamma') \leqslant 16^j G^{j+1}\sum_{Z_j,k_j}\cdots\sum_{Z_1,k_1}
 \ep{-\xi d(\gamma,Z_1)} \bigg( \prod_{i=2}^j \ep{-\xi d(Z_i,Z_{i-1})}  \bigg) \ep{-\xi d(\gamma',Z_j)}\bigg(\prod_{i=1}^j k_i^2f_{k_i}(Z_i)\bigg).
\end{equation*}
For $j=2$, we use first~(\ref{Assumption Interaction}) and then~(\ref{Assumption Interaction Point}) to get
\begin{align*}
a_2(\gamma,\gamma') 
&\leqslant 16^2 G^3 \sum_{Z_1,k_1}k_1^2f(Z_1)\ep{-\xi d(\gamma,Z_1)}\sum_{Z_2,k_2}k_2^2f(Z_2)\caD(Z_2)\ep{-\xi d(\gamma',Z_2)}\ep{-\xi d(Z_2,Z_1)} \\
&\leqslant  16^2 G^3 C_\Phi(\zeta,\xi) \sum_{Z_1,k_1}k_1^2f(Z_1)\ep{-\xi d(\gamma,Z_1)}\caD(Z_1)\ep{-\zeta d(\gamma',Z_1)}
\leqslant  16^2 G^3 C_\Phi(\zeta,\xi)^2\ep{-\zeta d(\gamma,\gamma')}.
\end{align*}
The same happens recursively for all $j$ and yields
\begin{equation*}
a_j(\gamma,\gamma') \leqslant 4^{2j} G^{j+1} C_\Phi(\zeta,\xi)^j \ep{-\zeta d(\gamma,\gamma')}.
\end{equation*}
With this, (\ref{LR Iteration}) (and (\ref{eq:eq:ineq_asum})) yields
\begin{equation*}
F_t^{\Lambda,\sharp}(\gamma,\gamma') 
\leqslant G\ep{-\xi d(\gamma,\gamma')} + G \sum_{j=1}^\infty \frac{(16GC_\Phi(\zeta,\xi) \vert t\vert)^j}{j!}\ep{-\zeta d(\gamma,\gamma')}
\leqslant G \ep{16 G C_\Phi(\zeta,\xi) \vert t\vert }\ep{-\zeta d(\gamma,\gamma')},
\end{equation*}
as we had set to prove.
\end{proof}

\subsection{Existence of the dynamics}

One of the classical applications of the Lieb-Robinson bound is the proof that the finite-volume dynamics $\tau_t^\Lambda$ converges to an infinite-volume dynamics $\tau_t$ as $\Lambda\to\Gamma$, in the strong topology of the operator algebra. 
\begin{cor}\label{cor: existence of dynamics}
Let $\{\Lambda_n\}_{n\in\bbN}$ be an increasing and absorbing sequence of $\Gamma$. Under the assumptions of Theorem~\ref{thm: LR bound}, for any $A\in\caA$, the sequence $\tau_t^{\Lambda_n}(A)$ is convergent to $\tau_t(A)$ and the family $\{\tau_t:t\in\bbR\}$ defines a strongly continuous one-parameter group of *-automorphisms of $\caA$.
\end{cor}

\begin{proof}
Setting $\gamma'=\gamma$ in~\eqref{Assumption Interaction Point} and choosing positive $\xi,\zeta'$ such that $\zeta := \xi +\zeta'<\xi_0$, one infers that
\begin{equation}\label{Existence Additional Hyp}
\sum_{Z\subset\Gamma}\sum_{k=1}^{|Z|} k f_k(Z)\ep{-\zeta d(\gamma,Z)}<\infty\;.
\end{equation}
Now, let $n\geq m$. Then for any $\gamma\in\Gamma$,
\begin{equation*}
\tau_t^{\Lambda_n}(a_\gamma) - \tau_t^{\Lambda_m}(a_\gamma) 
= \left.\tau_{s}^{\Lambda_n}\circ\tau_{t-s}^{\Lambda_m}(a_\gamma)\right\vert_{s=0}^{s=t}
= \iu \int_0^t\tau_s^{\Lambda_n}\left(\left[H_{\Lambda_n} - H_{\Lambda_m},\tau_{t-s}^{\Lambda_m}(a_\gamma)\right]\right)ds.
\end{equation*}
Taking the norm yields the estimate
\begin{equation*}
\left\Vert\tau_t^{\Lambda_n}(a_\gamma) - \tau_t^{\Lambda_m}(a_\gamma)\right\Vert
\leqslant \sum_{\substack{Z\subset\Lambda_n,\\Z\cap\Lambda_m^c\neq\emptyset}} \int_{0}^{\vert t\vert} \left \Vert [\Phi(Z),\tau_{s}^{\Lambda_m}(a_\gamma)]\right\Vert ds.
\end{equation*}
We now expand $\Phi(Z)$, and hence $M_k(Z)$ for all $k$, as in~(\ref{expand Mk}) to obtain
\begin{align*}
\left\Vert\tau_t^{\Lambda_n}(a_\gamma) - \tau_t^{\Lambda_m}(a_\gamma)\right\Vert
& \leqslant  2G\left(\ep{\zeta v \vert t\vert } - 1\right)\sum_{\substack{Z\subset\Lambda_n,\\Z\cap\Lambda_m^c\neq\emptyset}}  \sum_{k=1}^{|Z|} k f_k(Z) \ep{-\zeta d(\gamma,Z)}
\end{align*}
by Theorem~\ref{thm: LR bound}. This converges to zero by~(\ref{Existence Additional Hyp}). It follows that $\tau_t^{\Lambda_n}(a_\gamma)$ is a Cauchy sequence, therefore converging to $\tau_t(a_\gamma)$. Since $\tau_t^{\Lambda_n}$ is an automorphism, the map $\tau_t$ is well-defined on polynomials in $a_\gamma,a_\gamma^*$ and therefore extends to an automorphism of the whole algebra $\caA$ by continuity, see Lemma~\ref{lem:density}. 
\end{proof}

\subsection{On interactions}

A typical interaction we have in mind, here for particles without internal degree of freedom, is a two-body quartic interaction induced by a pair potential $W:\R^d\times\R^d\to\R$. Instead of introducing a UV regularization of the type used in~\cite{Gebert_LR}, we will make the assumption that there exists a continuous sets of states $\{\chi_x\}_{x\in\R^d}$ which is complete in $\s{H}$ and such that the discrete family $\{\chi_\gamma\}_{\gamma\in\Gamma}$, with $\Gamma\subset\R^d$, is a lattice-localized frame in $\s{H}$. In this sense the wavefunction~$\chi_x$ is interpreted as that of a particle localized around $x$ and the operators $a_x^\sharp:=a^\sharp(\chi_x)$ create on annihilate particle in the state $\chi_x$. We will see in Section  \ref{sec: Laudau Hamiltonian} that this picture fits perfectly for the description of charged particles in a plane subjected to a uniform perpendicular magnetic field. 
Formally, in second quantization, this leads to  interaction of the type
\begin{equation*}
I_W = \int_{\R^d\times\R^d}\dd x\dd y\; W(x,y)n_x n_{y},
\end{equation*}
where $n_x = a^*(\chi_x) a(\chi_x)$. Let now expand  every  $\chi_x$ on the frame $\{\chi_\gamma\}_{\gamma\in\Gamma}$ to get
\begin{equation*}
 a(\chi_x) = \sum_{\gamma\in\Gamma} \overline{s_\gamma(x)} a_\gamma
\end{equation*}
where $s_\gamma(x):=\langle S^{-1}\chi_\gamma,\chi_x\rangle$ as in~(\ref{Expansion in frame annihilation}). With this one gets (at least formally)
\begin{equation}\label{Density-density interaction}
\begin{aligned}
I_W\;&=\;\sum_{\gamma_1,\gamma_2,\gamma_3,\gamma_4\in\Gamma}w(\gamma_1,\gamma_2,\gamma_3,\gamma_4)
a\str_{\gamma_4}a_{\gamma_3}a\str _{\gamma_2}a_{\gamma_1}
\end{aligned}
\end{equation}
where
\begin{equation}\label{eq_int_w}
w(\gamma_1,\gamma_2,\gamma_3,\gamma_4)\; =\; \int_{\R^d\times\R^d}\dd x\dd y\; W(x,y) \; s_{\gamma_4}(x)\overline{s_{\gamma_3}(x)}s_{\gamma_2}(y)\overline{s_{\gamma_1}(y)}\;.
\end{equation}
We shall take here the right hand side of~(\ref{Density-density interaction}) as the definition of the interaction. It is of the form~(\ref{interaction}) and we shall assume that~(\ref{Assumption Interaction}). Formally, the decay properties of the coefficients $w(\gamma_1,\gamma_2,\gamma_3,\gamma_4)$ depend on the potential $W(x,y)$ and on the functions $s_{\gamma}(x)$.

\subsection*{The case of Landau systems}
Let us provide here a more precise analysis in the case of the Landau systems  using the information provided in  Sections \ref{sec:MagGabor} and  \ref{sec: Laudau Hamiltonian}. Recall that the space dimension is $d=2$. From the reproducing kernel property~(\ref{Reproducing kernel}), we have that
\[
\overline{s_\gamma(x)}\;=\;\frac{1}{\omega}\sqrt{\frac{\pi}{2}}(S^{-1}\chi_\gamma)(x)\;.
\]
Let us introduce the notation $v_\omega:=S^{-1}\chi_0=S^{-1}g_\omega$. By using the facts that $\chi_\gamma$ is related to $\chi_0$ by a magnetic translation and  that $S$ commutes with the magnetic translations one gets that
$S^{-1}\chi_\gamma=T_{2\omega\gamma} v_\omega$ and in turn
\[
\overline{s_\gamma(x)}\;=\;\frac{1}{\omega}\sqrt{\frac{\pi}{2}}\expo{-\ii 2\omega^2(\gamma\wedge x)} v_\omega\left(x-\gamma\right)\;.
\]
With this, (\ref{eq_int_w}) becomes
\[
w(\gamma_1,\gamma_2,\gamma_3,\gamma_4)\; =\;\left(\frac{\pi}{2\omega^2}\right)^2 \int_{\R^2\times\R^2}\dd x\dd y\; \expo{-\ii 2\omega^2[(\gamma_3-\gamma_4)\wedge x]}\expo{-\ii 2\omega^2[(\gamma_1-\gamma_2)\wedge y]}F^\omega_{\underline{\gamma}}(x,y)
\]
where
\[
F^\omega_{\underline{\gamma}}(x,y)\;:=\;W(x,y)\overline{v_\omega\left(x-\gamma_4\right)}v_\omega\left(x-\gamma_3\right)\overline{v_\omega\left(y-\gamma_2\right)}v_\omega\left(y-\gamma_1\right)\;.
\]
where we used the short notations $\underline{\gamma}:=(\gamma_1,\gamma_2,\gamma_3,\gamma_4)$. We now show that the kernel $w$ decays exponentially; We refer to~\cite{DelPrete} for the validity of Assumption~(ii).

\begin{prop}
Assume that
\begin{enumerate}
\item $W(x,y)=W(|x-y|)$ and $|W(x,y)|\leqslant C_1\expo{-\sigma_1|x-y|}$,
\item $|v_\omega(x)|<C_2\expo{-\sigma_2|x|}$.
\end{enumerate}
Then
\begin{equation*}
|w(\gamma_1,\gamma_2,\gamma_3,\gamma_4)| \leqslant 
K_\sigma \expo{-\sigma\mathrm{diam}\{\gamma_1,\gamma_2,\gamma_3,\gamma_4\}},
\end{equation*}
where $\sigma = \min\{\frac{\sigma_1}{2},\frac{\sigma_2}{6}\}$ and $K_\sigma = \frac{\pi^4 C_1C_2^4}{4\omega^4\sigma^4}$.
\end{prop}
\begin{proof}
After the change of variables $x\mapsto x+\gamma_3$ and $y\mapsto y+\gamma_1$ one gets
\[
w(\gamma_1,\gamma_2,\gamma_3,\gamma_4)\; =\;\left(\frac{\pi\phi}{2\omega^2}\right)^2 \int_{\R^2\times\R^2}\dd x\dd y\; \expo{-\ii 2\omega^2[(\gamma_3-\gamma_4)\wedge x]}\expo{-\ii 2\omega^2[(\gamma_1-\gamma_2)\wedge y]}\widetilde{F}_{\underline{\gamma}}(x,y)
\]
where $\phi\;:=\;\expo{\ii 2\omega^2[(\gamma_4\wedge\gamma_3)+(\gamma_2\wedge\gamma_1)]}$ is a phase and
\[
\widetilde{F}_{\underline{\gamma}}(x,y)
\;:=\;W\big(|(x-y)+(\gamma_3-\gamma_1)|\big)f_{\gamma_4-\gamma_3}(x)f_{\gamma_2-\gamma_1}(y)\]
with
\[
f_{\gamma}(x)\;:=\;\overline{v_\omega\left(x-\gamma\right)}
v_\omega\left(x\right)\;.
\]
Since $v_\omega$ is a rapidly decaying then the function $f_{\gamma}$ decays fast when $|\gamma|\to \infty$. More precisely, Assumption~(ii) implies that
\[
|f_{\gamma}(x)|\;<\;C_2^2\expo{-\sigma_2(|x-\gamma|+|x|)}\;<\;C_2^2\expo{-\frac{\sigma_2}{2}|\gamma|}\expo{-\frac{\sigma_2}{2}|x|}\;
\]
where we used that
\begin{equation}\label{eq_AAS}
|\gamma|+|x|\leqslant |\gamma|+|x|+|x-\gamma| \leqslant 2(|x-\gamma|+|x|).
\end{equation}
This and (i) now imply that
\[
|\widetilde{F}_{\underline{\gamma}}(x,y)|\leqslant C_1C_2^4\expo{-\frac{\sigma_2}{2}|\gamma_4-\gamma_3|}
\expo{-\frac{\sigma_2}{2}|\gamma_2-\gamma_1|}\expo{-\sigma_1|(x-y)+(\gamma_3-\gamma_1)|}\expo{-\frac{\sigma_2}{2}|x|}\expo{-\frac{\sigma_2}{2}|y|}.
\]
By an application of \eqref{eq_AAS} (using $-\gamma$ instead of $\gamma$) one gets $|\gamma|+|x-y| \leqslant 2|(x-y)+\gamma|+|x-y|$ and hence
\begin{equation*}
|\gamma|+|x-y|+|x+y| \leq 2|(x-y)+\gamma| + 3(|x| + |y|).
\end{equation*}
Picking now $\tilde\sigma = \min\{\sigma_1,\frac{\sigma_2}{3}\}$, we conclude that 
\begin{equation*}
|\widetilde{F}_{\underline{\gamma}}(x,y)|\leqslant
C_1C_2^4\expo{-\frac{\sigma_2}{2}|\gamma_4-\gamma_3|}
\expo{-\frac{\sigma_2}{2}|\gamma_2-\gamma_1|}\expo{-\frac{\tilde\sigma}{2}(|\gamma_3-\gamma_1|+|x-y|+|x+y|)}\expo{-\frac{\sigma_2-3\tilde\sigma}{2}(|x|+|y|)}
\end{equation*}
The last factor is uniformly bounded. Moreover, $\frac{\sigma_2}{2}\geq \tilde\sigma\geq\frac{\tilde\sigma}{2}$ and so
\begin{equation}\label{eq-ineq-Mas}
|w(\gamma_1,\gamma_2,\gamma_3,\gamma_4)| \leqslant K_{\frac{\tilde\sigma}{2}} \expo{-\frac{\tilde\sigma}{2}\left(|\gamma_4-\gamma_3|+|\gamma_2-\gamma_1|+|\gamma_3-\gamma_1|\right)}
\end{equation}
where
\begin{equation*}
K_{\frac{\tilde\sigma}{2}}:= C_1C_2^4\left(\frac{\pi}{2\omega^2}\right)^2\int_{\R^2\times\R^2}\dd x\dd y \expo{-\frac{\tilde\sigma}{2}(|x-y|+|x+y|)}
=\frac{4 \pi^4 C_1C_2^4}{\omega^4\tilde\sigma^4}.
\end{equation*}
Inequality \eqref{eq-ineq-Mas} proves yields the claim (with $\sigma = \frac{\tilde\sigma}{2}$) if the diameter is realized on any of the pairs $(\gamma_1,\gamma_2)$, $(\gamma_1,\gamma_3)$ or $(\gamma_3,\gamma_4)$. The same decay can be obtained from \eqref{eq-ineq-Mas} for any other pair by the triangle inequality.
\end{proof}

\section{Second quantization realized on localized frames}   \label{sec: Fock quantization}

In this section, we discuss how to use localized frames to express conveniently the second quantization of quadratic Hamiltonians in the GNS space of pure quasi-free states.

\subsection{Quasi-free second quantization}\label{sec: 2nd Q basis}

We first recall that a pure gauge-invariant quasi-free state $\omega_P$ over $\caA(\caH)$ is uniquely defined by an orthogonal projection $P = P\str = P^2$ on $\caH$ and the formula
\begin{equation}\label{quasi-free}
\omega_P(a\str(f_n)\cdots a\str(f_1)a(g_1)\cdots a(g_m)) = \delta_{nm}\det(\langle g_i, Pf_j\rangle_{i,j=1}^n).
\end{equation}
Let $a_0(f),a_0\str(f)$ be the creation and annihilation operators acting on the fermionic Fock space
\begin{equation*}
\rr{F}(\caH) := \bigoplus_{n=0}^{\infty}\caH\wedge\cdots\wedge \caH,
\end{equation*}
where $\caH\wedge\cdots\wedge \caH$ is the $n$-fold antisymmetric tensor product of $\caH$. Let $C$ be a conjugation on $\caH$ commuting with $P$. Then
\begin{equation*}
\pi_P(a(f)) := a_0(({\bf1}-P)f) + a_0\str(CPf)
\end{equation*}
form a representation of the CAR algebra on $\rr{F}(\caH)$ such that
\begin{equation*}
\langle\Omega, \pi_P(a\str (f)) \pi_P(a(g))\Omega\rangle = \langle\Omega, a_0 (CP f) a_0\str(CP g)\Omega\rangle = \langle CPf,CPg\rangle,
=\langle g,Pf\rangle
\end{equation*}
namely $(\rr{F}(\caH),\pi_P,\Omega)$ is a GNS representation of $\omega_P$. The case $P=0$ corresponds to the usual Fock state, while a non-trivial $P$ allows for a non-zero particle density: Indeed, if $n_P(f) := a_P\str (f) a_P(f)$, where we denote $a_P(f) := \pi_P(a(f))$, then
\begin{equation*}
\langle\Omega, n_P(f)\Omega\rangle = 1
\end{equation*}
for any normalized $f = Pf$. In the upcoming example, this would typically be a fractionally filled Landau level.

Second quantization of a self-adjoint operator $H$ on $\caH$ is a self-adjoint operator $\dd \Gamma_P(H)$ on $\rr{F}(\caH)$ such that
\begin{equation}\label{e of dGamma}
\ep{\iu t \dd\Gamma_P(H)} a_P^\sharp(f) \ep{-\iu t \dd\Gamma_P(H)} = a_P^\sharp(\ep{\iu t H}f),
\end{equation}
namely the corresponding unitary group $\ep{\iu t \dd\Gamma_P(H)}$ implements the automorphism $a(f)\mapsto a(\ep{\iu t H}f)$ in the representation $\pi_P$. The existence of this operator can be proved as follows, see~\cite{Lundberg} for details and proofs.

Let $L(t) = P^\perp \ep{\iu t H} P^\perp$ and $M(t)=CP\ep{\iu t H} P^\perp$, where we deonted $P^\perp = {\bf 1}-P$. If $\mathrm{tr}(PHP^\perp H)$ is finite, then $M(t)$ is a Hilbert-Schmidt operator and so is $K(t) := M(t) L(t)^{-1}$. Therefore, there is a square integrable sequence $\{\lambda_n\}_{n=1}^\infty$, an orthonormal basis $v_n(t)$ of $P\caH$ and an orthonormal basis $u_n(t)$ of $P^\perp\caH$ such that 
\begin{equation*}
K(t)f = \sum_n\lambda_n(t) v_n(t)\langle u_n(t),g\rangle.
\end{equation*}
With this, we define the vectors
\begin{equation*}
\Omega_P^N(t) := \ep{-\sum_{n=1}^N\lambda_n(t) a_0\str(u_n(t))a_0\str(v_n(t))}\Omega.
\end{equation*}
The fact that the sequence $\lambda_n(t)$ is square integrable implies that $\Omega_P^N(t)$ is strongly convergent and we define
\begin{equation*}
\Omega_P(t) := \caN(t) \ep{-\sum_{n=1}^\infty\lambda_n(t) a_0\str(u_n(t))a_0\str(v_n(t))}\Omega
\end{equation*}
where
\begin{equation*}
\caN(t) = \prod_{n=1}^\infty\left(1+\lambda_n(t)^2\right)^{-\frac{1}{2}}
\end{equation*}
ensures that $\Omega_P(s)$ is normalized. Then we have
\begin{thm}\label{thm:Bogolubov}
Assume that $\mathrm{tr}(PHP^\perp H)$ is finite and let $\Gamma_P(\ep{\iu t H})$ be defined by
\begin{align*}
 \Gamma_P(\ep{\iu t H})\Omega &= \Omega_P(t),\\
 \Gamma_P(\ep{\iu t H})a_P^\sharp(f)\Omega &= a_P^\sharp(\ep{\iu t H}f)\Omega_P(t).
\end{align*}
Then $\Gamma_P(\ep{\iu t H})$ extends to a unitary group on $\rr{F}(\caH)$ such that
\begin{equation*}
a_P(\ep{\iu t H}f) = \Gamma_P(\ep{\iu t H}) a_P(f) \Gamma_P(\ep{\iu t H})\str.
\end{equation*}
\end{thm}

\noindent With this, Stone's theorem ensures the existence of the generator $\dd\Gamma_P(H)$.

A particularly simple situation arises when $P$ is a spectral projection of $H$, in which case the assumption of the Theorem~\ref{thm:Bogolubov} are trivially satisfied. The definition~(\ref{quasi-free}) implies that $\omega_P$ is invariant under the dynamics, and the construction above yields indeed that $\Omega_P(t) = \Omega$ is constant. Let $\s{D}(H)$ be the domain of the self-adjoint operator $H$, and let $\{\psi_n\}_{n=1}^\infty$ be an orthonormal basis such that $\psi_m\in\s{D}(H)$ for all $m\in\bbN$. The formal expression
\begin{equation}\label{P Second Quant}
\dd\tilde\Gamma_P(H) := \sum_{m} a_P\str(H\psi_m)a_P(\psi_m)
\end{equation}
is such that 
\begin{equation*}
\iu[\dd\tilde\Gamma_P(H),a_P(f)] = a_P(\iu H f)
\end{equation*}
by the anticommutation relations and so (\ref{e of dGamma}) holds. We check that
\begin{equation*}
\dd\tilde\Gamma_P(H)\Omega = \sum_m \left(a_0\str(P^\perp H\psi_m)a_0\str(CP\psi_m) + \langle \psi_m,PHP\psi_m\rangle \right)\Omega
= \Tr(PH)\Omega.
\end{equation*}
where we used that picked $P\psi_m = 0$ or $P^\perp\psi_m = 0$ for all $m\in\bbN$ to conclude that the first term vanishes. Hence $\Omega$ is indeed invariant under the dynamics. This formal calculation exhibits the potential issue arising with~(\ref{P Second Quant}): If $PH$ is not trace class, then $\dd\tilde\Gamma_P(H)$ is not well defined. 

If $H^{(n)} = \sum_{j=1}^n {\bf 1}\otimes \cdots \otimes H \otimes \cdots\otimes{\bf 1}$, then the direct sum of the $H^{(n)}$ is essentially self-adjoint on the dense domain
\begin{equation}
\rr{D}_0(H) := \bigcup_{N\in\N} \left(\bigoplus_{n=0}^{N}\s{D}(H)_-^{(n)}\right)
\end{equation}
where $\s{D}(H)_-^{(n)}:=\s{D}(H)\wedge\ldots\wedge\s{D}(H)$ denotes the totally antisymmetric subspace of the $n-$fold tensor product of $\s{D}(H)$ with the usual convention $\s{D}(H)_-^{(0)}:=\C$, see~\cite{BratRob2}. The operator
\begin{equation}\label{Approx 2Q}
\dd\Gamma^M_P(H) : = \sum_{m=1}^M \left(a_P\str(H\psi_m)a_P(\psi_m) - \langle \psi_m,PHP\psi_m\rangle \right)
\end{equation}
is bounded operators for any $M\in\bbN$, and $\dd\Gamma^M_P(H)\Omega = 0$ by construction. Moreover,
\begin{equation*}
\dd\Gamma^M_P(H)a_P^\sharp(f_n)\cdots a_P^\sharp(f_1)\Omega = - \iu \sum_{j=1}^N a_P^\sharp(f_n)\cdots a_P^\sharp(\iu H f_j)\cdots a_P^\sharp(f_1)\Omega.
\end{equation*}
Since
\begin{equation*}
a_0\str(f)\Omega = \left(a_P\str(f) + a_P(Cf)\right)\Omega
\end{equation*}
and 
\[
\s{D}(H)_-^{(n)} = \overline{\mathrm{span}\big\{a_0\str(f_n)\cdots a_0\str(f_1)\Omega\;|\; f_j\in \s{D}(H)\text{ for all }j\in\{1,\ldots,n\}\big\}},
\]
 we conclude that $\dd\Gamma_P(H) :=  \mathrm{s}\!-\!\lim_{M\to\infty}\dd\Gamma^M_P(H)$ exists and is essentially self-adjoint on $\rr{D}_0(H)$. Therefore, as a consequence of \cite[Theorem VII.25]{Reed1980} one concludes that 
$\dd\Gamma_P^M(H)\to \dd\Gamma_P(H)$ in the strong resolvent sense. It then follows that
\begin{equation*}
\lim_{M\to\infty}\ep{\iu t \dd\Gamma_P^M(H)}\Phi = \ep{\iu t \dd\Gamma_P(H)}\Phi
\end{equation*}
for all $\Phi\in\rr{F}(\caH)$, see \cite[Theorem VIII.20]{Reed1980}.

\begin{rem}\label{Rk_exp}
Let us point out that the main point to prove the convergence of $\dd\Gamma_P^M(H)\Phi$ to $\dd\Gamma_P(H)\Phi$ is the continuity of the maps
\[
\s{H}\ni f\;\longmapsto a^\sharp_P(f)\phi^{(n)}\;\in\; \rr{F}(\caH)
\]
for every fixed $\phi^{(n)}\in \caH^{(n)}_-$ and $n\in\N$.
This property follows directly from the bound $\Vert a_P^\sharp (f)\phi^{(n)}\Vert\leqslant\Vert f \Vert \Vert \phi^{(n)}\Vert$, which follows from the anticommutation equations. The main implication of this fact is that if $f=\sum_{m\in\N}c_k\psi_m$
in the topology of $\s{H}$, then the expressions $a_P(f)\phi^{(n)}=\sum_{m\in\N}\overline{c_m}a_P(\psi_m)\phi^{(n)}$ and $a_P\str(f)\phi^{(n)}=\sum_{m\in\N}c_ma_P\str(\psi_m)\phi^{(n)}$ make sense in the topology of $\rr{F}(\caH)$ for every  $\phi^{(n)}\in \caH^{(n)}_-$. Said differently, the equalities
\[
a_P(f)=\sum_{m\in\N}\overline{c_m}a_P(\psi_m)\;,\qquad a_P\str(f)=\sum_{m\in\N}c_ma_P\str(\psi_m)
\]
make sense in the strong operator topology. It is worth noting that the family $\{\psi_m\}$ here is not  supposed to be an orthonormal basis of $\s{H}$. \hfill $\blacktriangleleft$
\end{rem}

\subsection{Second quantization based on lattice-localized frames}   \label{sub: CAR_frame}

We now extend the above to cover the case of frames instead of orthonormal bases to represent the operator $\dd\Gamma_P(H)$. Let $\{\chi_\gamma\}_{\gamma\in\Gamma}\subset\s{D}(H)$ be a frame; In order to not overburden the notation, we denote $\rr{a}_\gamma:=a_P(\chi_\gamma)$ for all $\gamma\in\Gamma$. For every $m\in\N$, the element $\psi_m$ of the orthonormal basis used in the representation \eqref{P Second Quant} can be expressed in terms of the frame as
\[
\psi_m = \sum_{\gamma\in\Gamma}s_\gamma(m)\chi_\gamma
\]
where $s_\gamma(m):=\langle S^{-1}\chi_\gamma,\psi_m\rangle$ and the bounded invertible operator $S$ is the frame operator \ref{eq:fr_op}. 
By using the same argument in Remark \ref{Rk_exp} one obtains
\begin{equation}\label{eq:ser_a_a*}
\begin{aligned}
a_P(\psi_m) = \sum_{\gamma\in\Gamma}\overline{s_\gamma(m)}\rr{a}_\gamma\;,\qquad
a_P\str(\psi_m) = \sum_{\gamma\in\Gamma}{s_\gamma(m)}\rr{a}^*_\gamma\end{aligned}
\end{equation}
where the convergence of the sums must be interpreted in the in the strong operator topology on Fock space. Moreover  
the convergence is unconditional meaning that the order in which the set $\Gamma$ is spanned is irrelevant. In the same way, 
\[
a_P\str(\psi_n)a_P(\psi_m) = \sum_{\gamma,\gamma'\in\Gamma}{s_{\gamma'}(n)}\overline{s_\gamma(m)} \rr{a}^*_{\gamma'}\rr{a}_\gamma\;,
\]
where again the series converges unconditionally with respect to the strong operator topology.
Plugging this formula into~\eqref{Approx 2Q} one gets
\begin{equation}\label{SQ_bas_3}
\dd\Gamma_P^M(H)=\sum_{\gamma,\gamma'\in\Gamma}\bigg(\sum_{m=1}^M \langle S^{-1}\chi_{\gamma'}, H\psi_m\rangle \overline{s_\gamma(m)} \rr{a}^*_{\gamma'}\rr{a}_\gamma - {s_{\gamma'}(m)}\overline{s_\gamma(m)}\langle \chi_\gamma,PHP\chi_{\gamma'}\rangle \bigg)
\end{equation}
which express $\dd\Gamma_P^M(H)$ as a strong limit of the $\rr{a}^*_{\gamma'}\rr{a}_\gamma$.


Let us study now the quantities
\[
t_H(\gamma',\gamma) := \sum_{m\in\N}\langle S^{-1}\chi_{\gamma'}, H\psi_m\rangle \overline{s_\gamma(m)},\qquad 
c_H(\gamma) := \sum_{\gamma'\in\Gamma}\sum_{m\in\N}{s_{\gamma'}(m)}\overline{s_\gamma(m)}\langle \chi_\gamma,PHP\chi_{\gamma'}\rangle.
\]
\begin{lemma}\label{lma: hopping terms}
Let $H$ be a self-adjoint operator and assume that $S^{-1}\chi_\gamma\in\s{D}(H)$ for every $\gamma\in\Gamma$.
Then, it holds true that
\begin{equation*}
t_H(\gamma,\gamma') = \left\langle \chi_{\gamma'},(S^{-1})HS^{-1} \chi_\gamma\right\rangle,\qquad
c_H(\gamma) = \langle \chi_\gamma, PHP S^{-1}\chi_\gamma\rangle.
\end{equation*}
Under the more restrictive hypothesis of 
 Corollary \ref{Cor_1}  one obtains the estimate
\[
|t_H(\gamma',\gamma)|\;\leqslant\; a_2 \delta_{E(\gamma),E(\gamma')} \; E(\gamma) \expo{-\lambda_2 d(\gamma,\gamma')}\;,\qquad \gamma,\gamma'\in\Gamma\;.
\]
\end{lemma}
\begin{proof}
If $t_H^M(\gamma',\gamma)$ denote the sum defining $t_H(\gamma',\gamma)$ truncated at $m=M$, we see that
\begin{equation*}
t_H^M(\gamma',\gamma) = \sum_{m=1}^M\langle S^{-1}\chi_{\gamma'}, H\psi_m\rangle\langle \psi_m, S^{-1}\chi_\gamma\rangle = \langle H S^{-1}, Q_MS^{-1}\chi_\gamma\rangle
\end{equation*}
where $Q_M:=\sum_{m=1}^M\ketbra{\psi_m}{\psi_m}$ is the projection on the first $M$ element of the orthogonal basis. The first claim follows by observing that $Q_M$ converges strongly, and in turn weakly, to the identity, together with the self-adjointness of $S$ and $H$. The same limiting procedure yields that
\begin{equation*}
c_H(\gamma) = \sum_{\gamma'\in\Gamma}\langle \chi_\gamma,PHP\chi_{\gamma'}\rangle\langle S^{-1}\chi_{\gamma'}, S^{-1}\chi_\gamma\rangle
=\langle \chi_\gamma,PHP S^{-1}\chi_\gamma\rangle
\end{equation*}
where we used~(\ref{eq:frame_basis}) in the second equality.
\end{proof}

We observe that~(\ref{eq:frame_basis}) also yields that
\begin{equation*}
\sum_{\gamma'\in\Gamma}t_H(\gamma,\gamma')\rr{a}^*_{\gamma'}\rr{a}_\gamma = a_P\str(HS^{-1}\chi_\gamma)a_P(\chi_\gamma).
\end{equation*}

Let us now consider  way of approximating $\dd\Gamma_P(H)$ different from \eqref{SQ_bas_3}. Let $\s{P}_f(\Gamma)$ be the collection of finite subset of $\Gamma$ and for each $\Lambda\in\s{P}_f(\Gamma)$ let us consider the bounded operator 
\[
\dd\Gamma_P^\Lambda(H):=\sum_{\gamma\in\Lambda} \left(a_P\str(HS^{-1}\chi_\gamma)a_P(\chi_\gamma) - c_H(\gamma)\right).
\]

\begin{prop}\label{prop_conv_gab}
Let $H$ be a self-adjoint operator and assume that $S^{-1}\chi_\gamma\in\s{D}(H)$ for every $\gamma\in\Gamma$.
Then,  $\dd\Gamma_P^\Lambda(H)$ converges in the {strong resolvent} sense to $\dd\Gamma_P(H)$ when $\Lambda\nearrow \Gamma$.
\end{prop}
\begin{proof}
The proof follows the same strategy described at the end of Section \ref{sec: 2nd Q basis} and uses the property described in  Remark \ref{Rk_exp}. First of all one checks that $\dd\Gamma_P^\Lambda(H)\to \dd\Gamma_P(H)$ for every $\Phi\in \rr{D}_0(H)$ when $\Lambda\nearrow \Gamma$. The result follows again from \cite[Theorem VII.25]{Reed1980}.
\end{proof}

The last result allows us to represent the second quantization of $H$ as the series
\[
\dd\Gamma(H) = \lim_{\Lambda\nearrow \Gamma}\bigg(\sum_{\gamma,\gamma'\in\Lambda} t_H(\gamma',\gamma) \rr{a}^*_{\gamma'}\rr{a}_\gamma - \sum_{\gamma\in\Lambda}c_H(\gamma)\bigg)
\]
which is convergent in strong resolvent sense. We point out here that the expression in the right hand side, seen at the level of the abstract algebra, corresponds to a bonafide `interaction' in the sense of Section~\ref{sub: Hamiltonian}, and under the assumptions of Corollary~\ref{Cor_1} through Lemma~\ref{lma: hopping terms}, this interaction satisfies Assumption~\ref{Hyp:Interaction}.

\begin{ex}[Landau Hamiltonian]\label{Rk_land_ham}
The case of the second quantization of the Landau Hamiltonian can now be discussed. First of all from $L^2(\R^2)=\bigoplus_{r\in\N_0}\s{H}_{r}$ one infers that 
\[
\rr{F}_-(L^2(\R^2)) := \bigoplus_{r\in\N}\rr{F}_-(\s{H}_{r})
\]
for the structure of the Fock space associated to $L^2(\R^2)$. Secondly, the Landau Hamiltonian reads $H_B=\bigoplus_{r\in\N_0}q(r)\Pi_r$ where $\Pi_r$ is the orthogonal projection onto  $\s{H}_{r}$ and $q(r)=r+\frac{1}{2}$. We conclude that 
\[
\dd\Gamma_P(H_B) = \bigoplus_{r\in\N_0}q(r)\dd\Gamma_P(\Pi_r)\;.
\]
By using the description in Section~\ref{sub: CAR_frame}, we conclude that
\[
\dd\Gamma_P(H_B) = \sum_{r\in\N_0}\Big(\sum_{\gamma,\gamma'\in\Gamma}t_r(\gamma',\gamma) \rr{a}^*_{r,\gamma'}\rr{a}_{r,\gamma} - \sum_{\gamma\in\Lambda}c_r(\gamma)\Big)
\]
where
\begin{align*}
t_r(\gamma',\gamma) &= q(r)\left\langle \chi_{(r,\gamma')},S^{-2}\chi_{(r,\gamma)}\right\rangle,\\
c_r(\gamma) &= \langle \chi_{(r,\gamma)}, S^{-1}\chi_{(r,\gamma)}\rangle
\end{align*}
and $\rr{a}^\sharp_{r,\gamma}:=\rr{a}^\sharp(\chi_{(r,\gamma)})$. The decay of the coefficients is exponential in $d(\gamma,\gamma')$ as per Proposition~\ref{Prop: localization of S^-p}.  \hfill $\blacktriangleleft$
\end{ex}

\section{Interacting particles in a transverse magnetic field}   \label{sec: Laudau Hamiltonian}

We finally turn to the concrete description of the physically relevant case in which a frame is adapted to a concrete physical system, namely that of a two-dimensional electron gas subject to a perpendicular magnetic field, as in the setting of the quantum Hall effect. Here, the lattice-localized frame arises naturally as being made of eigenstates of the Landau Hamiltonian that are exponentially localized (the decay is in fact Gaussian) because of the cyclotron orbits. The introduction of electron-electron interactions in this setting, which is essential for the theoretical understanding of the fractional quantum Hall effect, is naturally carried out in the framework we proposed above. 

The results of this section are not new, but recall them to show how they fit the general setting. The results of the previous sections then imply that an interacting version of the Landau Hamiltonian defines a proper dynamics satisfying a Lieb-Robinson bound. 

In presence of a magnetic field, the quantum dynamics of a particle of {mass} $m$   and {charge} $q$   is generated  by the 
 magnetic Schr\"{o}dinger operator %
\begin{equation}\label{eq:mSo1}
H_A := \frac{1}{2m}\left(-\ii \hslash \nabla-\frac{q}{c}\ A\right)^2
\end{equation}
acting on the Hilbert space $L^2(\R^d)$, where $c$ is the speed of  light and $\hslash$ is the Planck constant. The vector potential $A$ (1-form) is responsible for 
the coupling of the particle with the magnetic field $B:=\dd A$ (2-form). Under quite general assumptions on the vector potential $A$ the operator \eqref{eq:mSo1}
turns out to be essentially self-adjoint on the cores ${C}_c(\R^d)\subset {S}(\R^2)$ given by the compactly supported continuous functions and the Schwartz functions, respectively \cite[Theorem 3]{leinfelder-simader-81}.

\subsection{The Landau Hamiltonian}\label{sub: Landau}
\label{subsec:LH}
The  Landau Hamiltonian  $H_B$ is the magnetic Schr\"{o}dinger operator on $L^2(\R^2)$ with vector potential describing a uniform perpendicular magnetic field. In this work we will use the symmetric gauge, \ie
\begin{equation}\label{eq:vp}
A^{\rm sym}_B(x) := \frac{B}{2}{\rm e}_\bot\wedge x = \frac{B}{2}(-x_2,x_1),\qquad x := (x_1,x_2)\in\R^2
\end{equation}
where ${\rm e}_\bot:=(0,0,1)$ is the unit vector orthogonal to the plan $\R^2$ where the particle is confined, and $B\in\R$ describes the strength (and the orientation with respect to ${\rm e}_\bot$) of the magnetic field. Then  $H_B:=H_{A^{\rm sym}_B}$  is  the two-dimensional {magnetic} Schr\"{o}dinger operator \eqref{eq:mSo1} with the vector potential \eqref{eq:vp}. In the following we will assume  $B>0$ which means that the magnetic field is positively oriented with respect to the direction of ${\rm e}_\bot$.

Henceforth we will assume $q<0$ which corresponds to the case of electrons. With the constants which appear in the definition of $H_A$ and $A^{\rm sym}_B$ one can define the magnetic energy and the magnetic length
$$
\epsilon_B := \frac{|q|B\hslash}{mc}\;,\qquad
\ell_B := \sqrt{\frac{c\hslash}{|q|B}}\;.
$$
With this, the Landau Hamiltonian can be written as
\begin{equation}\label{eq:LH}
H_B :=  \frac{\epsilon_B}{2}\left(K_{1}^2+K_{2}^2\right)
\end{equation}
where the magnetic kinetic momenta $K_{1}$ and $K_{2}$ are defined by
\begin{equation}\label{eq:Kms}
K_{1} := \left(-{\ii }\ell_B\frac{\partial}{\partial x_1}-\frac{x_2}{2\ell_B}\right),\qquad K_{2} := \left(-{\ii }\ell_B \frac{\partial}{\partial x_2}+\frac{x_1}{2\ell_B}\right)\;.
\end{equation}
The expressions  \eqref{eq:Kms}  define
essentially self-adjoint operators with core ${C}_c(\R^d)$ 
 and  we will use the same symbols  to denote   their unique self-adjoint extensions.

The spectral theory of the 
Landau Hamiltonian $H_B$ is  a well established topic  and it is closely related to the elementary theory
of the harmonic oscillator. In order to compute the spectrum of $H_B$ let us introduce the dual momenta
\begin{equation}\label{eq:Dms}
G_{1} := \left(-{\ii }\ell_B\frac{\partial}{\partial x_2}-\frac{x_1}{2\ell_B}\right),\qquad G_{2} := \left(-{\ii }\ell_B \frac{\partial}{\partial x_1}+\frac{x_2}{2\ell_B}\right)\;,
\end{equation}
which defines again a pair of self-adjoint operators with core ${C}_c(\R^2)$. Moreover, the  commutation relations 
\begin{equation}\label{eq:cc1}
\begin{aligned}
&[K_{1},K_{2}] =   -\ii  {\bf 1} = [G_{1},G_{2}]\\
&[K_{i},G_{j}] = 0\ \ \ \ \qquad i,j=1,2
\end{aligned}
\end{equation}
can be easily proved by a direct computation on the core ${C}_c(\R^2)$.
In view of \eqref{eq:cc1} one can define two pairs of creation-annihilation operators
\begin{equation}\label{eq:cr-an-op}
\rr{a}^{\pm} := \frac{1}{\sqrt{2}}\left(K_{1}\pm\ii K_{2}\right)\;,\qquad \rr{b}^{\pm}:=\frac{-1}{\sqrt{2}}\left(G_{1}\pm\ii G_{2}\right)\;.
\end{equation}
The $\rr{a}^{\pm}$ and $\rr{b}^{\pm}$, are closable operators initially defined on the dense domains ${C}_c(\R^2)$. 
Moreover $\rr{a}^{-}$ and $\rr{b}^{-}$ are the  adjoint of $\rr{a}^{+}$ and $\rr{b}^{+}$, respectively. The identities~(\ref{eq:cc1}) imply the following bosonic commutation rules: 
\begin{equation}\label{eq:cc2}
[\rr{a}^{\pm},\rr{b}^{\pm}] = 0\;, \qquad \qquad [\rr{a}^{-},\rr{a}^{+}] = {\bf 1} = [\rr{b}^{-},\rr{b}^{+}]\;.
\end{equation}

The function 
\begin{equation}\label{eq:herm1}
\psi_{0}(x) := \frac{1}{\ell_B\sqrt{2\pi}}\ \expo{-\frac{|x|^2}{4\ell_B^2}}\;\in\;{S}(\R^2)\;.
\end{equation}
is normalized and satisfies $\rr{a}^{-}\psi_{0}=0=\rr{b}^{-}\psi_{0}$. Acting on $\psi_{0}$ with the creation operators one defines 
\begin{equation}\label{eq:herm2}
\psi_{n} := \frac{1}{\sqrt{n_1!n_2!}}\ (\rr{a}^{+})^{n_1}(\rr{b}^{+})^{n_2}\psi_{0},\qquad n := (n_1,n_2)\in\N^2_0\;
\end{equation}
where $\N_0:=\{0\}\cup\N$.
One has that $\psi_{n}\in {S}(\R^2)$ for any $n\in\N^2_0$. Moreover, a recursive application of the commutation rules \eqref{eq:cc2} yields that these functions form an orthonormal set. In fact, $\{\psi_{n}\ |\ n\in\N^2_0\}\subset {S}(\R^2)$ is a complete orthonormal system for $L^2(\R^2)$ called magnetic Laguerre basis.  
The functions \eqref{eq:herm2} can be expressed as \cite{johnson-lippmann-49,raikov-warzel-02}
\begin{equation}\label{eq:lag_pol}
\psi_{n}(x) := \psi_0(x)\ \sqrt{\frac{n_1!}{n_2!}}\left[\frac{x_1-\ii x_2}{\ell_B\sqrt{2}}\right]^{n_2-n_1}L_{n_1}^{(n_2-n_1)}\left(\frac{|x|^2}{2\ell_B^2}\right)
\end{equation}
where
$$
L_m^{(\alpha)}\left(\zeta\right) := \sum_{j=0}^{m}\frac{(\alpha+m)(\alpha+m-1)\ldots(\alpha+j+1)}{j!(m-j)!}\left(-\zeta\right)^j\;,\quad\alpha,\zeta\in \R
$$
are the generalized Laguerre polynomials of degree $m$ (with the usual convention $0!=1$) \cite[Sect. 8.97]{gradshteyn-ryzhik-07}.  


By 
using the definitions \eqref{eq:cr-an-op} and the commutation relations \eqref{eq:cc1}, a standard calculation yields
\begin{equation}\label{eq:Lh_no}
H_B = \epsilon_B\left(\rr{a}^{+}\rr{a}^{-}+\frac{1}{2}{\bf 1}\right).
\end{equation}
The first consequence  is that any Laguerre vector $\psi_{n}$ is an eigenvector of $H_B$. This implies that the Laguerre  basis  provides an orthonormal system which diagonalizes $H_B$ according to
$$
H_B\psi_{n}\; =\; \epsilon_B\left(n_1+\frac{1}{2}\right)\psi_{n},\qquad\qquad n = (n_1,n_2)\in\N_0^2.
$$
Hence, the spectrum of $H_B$ is a sequence of eigenvalues given by
\begin{equation}\label{eq:spec_H_B}
\sigma(H_B)=\left.\left\{q(r):=\epsilon_B\left(r+\frac{1}{2}\right)\ \right| \ r\in\N_0\right\}\;.
\end{equation}
We refer to the eigenvalue  $q(r)$ as the $r$-th Landau level. 


Each Landau level  is infinitely degenerate. 
A simple computation shows that the orthogonal component of the angular moment can be written as
\begin{equation*}
L_3=-\ii\hslash\left(x_1\frac{\partial}{\partial x_2}-x_2\frac{\partial}{\partial x_1}\right)
 = \hslash\left(\rr{a}^{+}\rr{a}^{-}-\rr{b}^{+}\rr{b}^{-}\right)\;.
\end{equation*}
This  and~(\ref{eq:cc2}) impliy that $[H_B,L_3]=0$.
Then, the Laguerre functions $\psi_{n}$ are simultaneous  eigenfunctions of $H_B$ and $L_3$ with eigenvalues $q(n_1)$ and  $\hslash(n_1-n_2)$, respectively. This shows that the possible eigenvalues of the angular momentum $L_3$ for a particle in the energy level $n_1$  are $l_m:=\hslash m$ with $-\infty < m\leqslant n_1$.


Let $\s{H}_r\subset L^2(\R^2)$ be the eigenspace relative to the eigenvalue $q(r)$ of $H_B$. With a slight abuse of nomenclature we will call
$\s{H}_r$  the $r$-th {Landau level}.
The subspace   $\s{H}_r$ is spanned by $\psi_{(r,m)}$ with $m\in\N_0$ and the spectral projection $\Pi_{r}:L^2(\R^2)\to\s{H}_r$ is expressed in Dirac notation as
\begin{equation}\label{eq:Lan_proj}
\Pi_{r} := \sum_{m=0}^{\infty}\ketbra{\psi_{(r,m)}}{\psi_{(r,m)}}\;.
\end{equation}
One infers from \eqref{eq:herm2} the  recursive relations $\Pi_{r}=r^{-1} \rr{a}^{+} \Pi_{r-1}  \rr{a}^{-}$ and hence
$$
\Pi_{r} = \frac{1}{r!}\;(\rr{a}^{+})^r  \Pi_{0} (\rr{a}^{-})^r\;.
$$

\subsection{The {magnetic}  translations}
\label{subsect:magn_BFZ_tras}
With the help of the dual momenta \eqref{eq:Dms} one can build the family of unitary operators
\[
T_\gamma := \expo{-\ii(\gamma_1 G_2+\gamma_2 G_1)} = \expo{\frac{\ii}{2}\gamma_1\gamma_2}\expo{-\ii\gamma_1 G_2}\expo{-\ii\gamma_2 G_1}
\;,\qquad \gamma:=(\gamma_1,\gamma_2)\in \R^2\;.
\]
An explicit computation provides
$$
(T_\gamma\psi)(x) = \expo{-\ii\frac{\gamma\wedge x}{2\ell_B}}\psi\left(x-\ell_B\gamma\right)\;,\qquad\psi\in L^2(\R^2)
$$
where $\gamma\wedge x:=\gamma_1x_2-\gamma_2x_1$.
The operators $T_\gamma$
 are called magnetic  translations \cite{zak1,zak2}. As a consequence of the commutations relations \eqref{eq:cc1} one gets that
$$
[T_\gamma,H_B] = 0\;,\qquad\quad \forall\; \gamma\in\R^2
$$
namely the operators $T_\gamma$ are symmetries of the Landau Hamiltonian $H_B$. 
Moreover, from the definition, one can check that
$$
T_{\gamma+\gamma'} = \expo{\ii\frac{\gamma\wedge\gamma'}{2}}
T_\gamma T_{\gamma'}\;,\qquad\quad\forall\; \gamma,\gamma'\in\R^2\;. 
$$ 
This shows that the mapping $\gamma\mapsto T_\gamma$   provides a projective unitary representation of the group $\R^2$ on $L^2(\R^2)$ which leaves invariant $H_B$. To obtain a commutative representation one can restrict to a discrete subgroup of  $\R^2$.
For a given deformation parameter $\vartheta>0$
let us introduce the magnetic lattice 
\begin{equation}\label{avrth_lat}
\Gamma_\vartheta := \sqrt{2\pi}\vartheta\Z\;\times\;\frac{\sqrt{2\pi}}{\vartheta}\Z
\end{equation}
which is isomorphic to $\Z^2$. Since $\gamma\wedge\gamma'\in2\pi\Z$ for every $\gamma,\gamma'\in\Gamma_\vartheta$ one concludes that 
 $$
T_\gamma T_{\gamma'} = 
T_{\gamma'}T_\gamma \;,\qquad\quad\forall\; \gamma,\gamma'\in\Gamma_\vartheta\;. 
$$
It is worth noticing that the area of the unit cell of  the lattice $\Gamma_\vartheta$ is $4\pi$ independently of the deformation parameter.

\subsection{Magnetic Gabor frames for the lowest Landau level}\label{secop_fr_LLL}
We now connect the above well-known facts to the general setting of Section~\ref{sec: Gabor frames}. The lowest Landau level (LLL) is given by
\[
\s{H}_0 = \overline{{\rm span}\;\{\psi_{(0,m)}\;|\; m\in\N_0\}}\;,
\]
where
\begin{equation}\label{eq;psi_001}
\psi_{(0,m)}(x) = \frac{1}{\ell_B\sqrt{2\pi m!}}\left(\frac{x_1-\ii x_2}{\ell_B\sqrt{2}}\right)^m\; \expo{-\frac{|x|^2}{4\ell_B^2}}
\end{equation}
see~\eqref{eq:lag_pol}. In particular, we recover~\eqref{eq:descr_H_0} with $\omega=(2\ell_B)^{-1}$.


For every $\gamma\in\R^2$ we let
\begin{equation*}
\chi_\gamma:=T_{\ell_B^{-1}\gamma}\psi_{(0,0)}.
\end{equation*}
Since the magnetic translations commute with $H_B$, one automatically infers that $\chi_\gamma\in\s{H}_0$ for every $\gamma\in\R^2$.
An explicit computation provides
\[
\chi_\gamma(x) = \frac{1}{\ell_B\sqrt{2\pi}}\expo{-\ii\frac{\gamma\wedge x}{2\ell_B^2}}\;\expo{-\frac{\left|x-\gamma\right|^2}{4\ell_B^2}}\;,
\]
\ie the vectors $\chi_\gamma$ are exactly the vectors introduced in Section~\ref{sec:MagGabor}. It will be useful to expand the vector $\chi_\gamma$ on the magnetic Laguerre basis.

\begin{lemma}\label{lem020}
It holds that
\begin{equation}\label{eq;chi_002}
\chi_\gamma = \expo{-\frac{|\gamma|^2}{4\ell_B^2}}\sum_{m\in\N_0}\frac{1}{\sqrt{m!}}\left(\frac{\gamma_1+\ii \gamma_2}{\ell_B\sqrt{2}}\right)^m\psi_{(0,m)}\;,
\end{equation}
with the usual convention $0!=0^0=1$.
\end{lemma}
\begin{proof}
From the identity
\[
\begin{aligned}
|a-b|^2+\ii 2(a\wedge b)\;&=\;|a|^2+|b|^2-2[a\cdot b-\ii(a\wedge b)]\\
&=\;|a|^2+|b|^2-2(a_1+\ii a_2)(b_1-\ii b_2)
\end{aligned}
\]
valid for every $a,b\in\R^2$,
and the substitutions $a = \frac{\gamma}{2\ell_B}$ and $b =  \frac{x}{2\ell_B}$, one gets
\begin{equation}\label{eq;chi_001}
\chi_\gamma(x) = \frac{1}{\ell_B\sqrt{2\pi}} \expo{-\frac{|\gamma|^2}{4\ell_B^2}}e_\gamma(x){\expo{-\frac{|x|^2}{4\ell_B^2}}} = \expo{-\frac{|\gamma|^2}{4\ell_B^2}} e_\gamma(x)\psi_{(0,0)}(x),
\end{equation}
see~(\ref{eq:herm1}), where
\[
\begin{aligned}
e_\gamma(x)\;:&=\;\expo{\frac{\gamma_1+\ii \gamma_2}{2\ell_B^2}(x_1-\ii x_2)}
\;&=\;\sum_{m\in\N_0}\frac{1}{m!}\left(\frac{\gamma_1+\ii \gamma_2}{2\ell_B^2}\right)^m(x_1-\ii x_2)^m\;.
\end{aligned}
\]
The claim now follows by comparing this with~(\ref{eq;psi_001}).
\end{proof}


The identity~\eqref{eq;chi_002} implies that
\begin{equation}\label{eq:sca_prod}
\langle\chi_\gamma,\psi_{(0,m)} \rangle = \ell_B\sqrt{2\pi}\;\psi_{(0,m)}(\gamma)\;,
\end{equation}
for every $\gamma\in\R^2$ and $m\in\N_0$, and in turn the remarkable fact that
\begin{equation}\label{Reproducing kernel}
\langle\chi_\gamma,\phi \rangle = \ell_B\sqrt{2\pi}\;\phi(\gamma)\;,\qquad \forall\; \phi\in\s{H}_0
\end{equation}
namely, $\chi_\gamma$ define evaluation functionals at the point $\gamma$ in the LLL.

The vectors $\chi_\gamma$ are normalized by construction, but they are not orthogonal. However, as announced in Section~\ref{sec:MagGabor}, they are almost orthogonal.
\begin{prop}\label{prop:01}
For every $\gamma,\gamma'\in\R^2$ it holds true that
\[
\langle\chi_\gamma,\chi_{\gamma'} \rangle = \expo{\ii\frac{\gamma\wedge\gamma'}{2\ell_B^2}}\expo{-\frac{|\gamma-\gamma'|^2}{4\ell_B^2}}\;.
\]
\end{prop}
\begin{proof}
From the formula \eqref{eq:sca_prod} one infers that
\[
\langle\chi_\gamma,\psi_{(0,0)} \rangle =  \expo{-\frac{|\gamma|^2}{4\ell_B^2}}\;
\]
for every $\gamma\in\R^2$.
By construction one has that
\[
\langle\chi_\gamma,\chi_{\gamma'} \rangle\;=\:
\langle T_{\ell_B^{-1}\gamma}\psi_{(0,0)},T_{\ell_B^{-1}\gamma'}\psi_{(0,0)} \rangle\;=\:\langle T_{\ell_B^{-1}\gamma'}^*T_{\ell_B^{-1}\gamma}\psi_{(0,0)},\psi_{(0,0)} \rangle,
\]
and since the magnetic translations are a projective representation, 
\[
T_{\ell_B^{-1}\gamma'}^*T_{\ell_B^{-1}\gamma}\psi_{(0,0)} = \expo{-\ii\frac{\gamma\wedge\gamma'}{2\ell_B^2}}T_{\ell_B^{-1}(\gamma-\gamma')}\psi_{(0,0)} = \expo{-\ii\frac{\gamma\wedge\gamma'}{2\ell_B^2}}\chi_{\gamma-\gamma'}
\]
which concludes the proof.
\end{proof}


The family $\{\chi_\gamma\}_{\gamma\in\R^2}$ contains enough elements to span the full $\s{H}_0$. Said differently it allows for a resolution of the Landau projection $\Pi_0$ as
\begin{equation}\label{id_proy}
\Pi_0 = \frac{1}{2\pi\ell_B^2}\int_{\R^2}\dd^2\gamma\;\ketbra{\chi_\gamma}{\chi_\gamma}
\end{equation}
where $\ketbra{\chi_\gamma}{\chi_\gamma}$ denotes the projection on the state $\chi_\gamma$. The formula \eqref{id_proy} can be proved by using the expansion \eqref{eq;chi_002} to write the integral kernel of $\ketbra{\chi_\gamma}{\chi_\gamma}$ in terms of the Laguerre functions and then performing the integration over $\R^2$ to obtain the integral kernel of $\Pi_0$. Relation \eqref{id_proy} expresses the completeness of the family $\{\chi_\gamma\}_{\gamma\in\R^2}$: for every $\phi\in\s{H}_0$ such that $\Pi_0\phi=\phi$, 
\begin{equation}\label{eq:cont_expo}
\phi = \frac{1}{2\pi\ell_B^2}\int_{\R^2}\dd^2\gamma\;\langle\chi_\gamma,\phi\rangle \chi_\gamma.
\end{equation}
%


Since the vectors $\chi_\gamma$ are not linearly independent of each other, one may wonder whether it is possible to reduce \eqref{eq:cont_expo} to a discrete sum. For that let us choose positive constants $\alpha>0$ and $\beta>0$ and consider the lattice
\[
\Gamma_{(\alpha,\beta)} := \alpha\Z\times\beta\Z\;.
\]
The following theorem, which summarizes results in~\cite{bargmann-butera-girardello-klauder-71,perelomov-71,bacry-grossmann-zak-75} shows that the family $\{\chi_\gamma\}_{\gamma\in\Gamma_{(\alpha,\beta)}}$ is a complete set for $\s{H}_0$ provided the area $\alpha\beta$ is small enough. Let
$$\Delta_{\alpha,\beta}:=\frac{\alpha\beta}{\ell_B^2}.$$
\begin{thm}\label{th:LLL01}
(i) The family of vectors $\{\chi_\gamma\}_{\gamma\in\Gamma_{(\alpha,\beta)}}$ is a complete  set for $\s{H}_0$ if and only if $\Delta_{\alpha,\beta}\leqslant 2\pi$. In this case the family $\{\chi_\gamma\}_{\gamma\in\Gamma_{(\alpha,\beta)}}$ is overcomplete in the sense that any vector of the family is contained in the closed linear span of the other 
vectors. \\
\noindent (ii) When
$\Delta_{\alpha,\beta}< 2\pi$,
if one removes from $\{\chi_\gamma\}_{\gamma\in\Gamma_{(\alpha,\beta)}}$ any finite number of vectors then the remaining family is still overcomplete. In the threshold case $\Delta_{\alpha,\beta}= 2\pi$ the family is still complete by removing a single vector but becomes incomplete if any two vectors are removed.
\end{thm}

\noindent The threshold case  is known as a von Neumann lattice and it provides optimal (or minimal) complete families for $\s{H}_0$. 
Any such lattice is (up to a rescaling) of the form $\Gamma_\vartheta$ as given in \eqref{avrth_lat} and corresponds to the case of commuting magnetic translations.


Let us introduce the Jacobi's theta-function
\begin{equation*}
\theta_3(z|\tau):=\sum_{n\in\Z}\expo{\ii 2\pi z n}\;\expo{\ii \pi \tau n^2},
\end{equation*}
 defined for every $z\in\C$ whenever ${\rm Im}(\tau)>0$, and let
\begin{equation}\label{eq:B_bou}
M_{\alpha,\beta} := \theta_3\left(0\left|\ii\frac{\alpha^2}{4\pi\ell_B^2}\right)\right.\; \theta_3\left(0\left|\ii\frac{\beta^2}{4\pi\ell_B^2}\right)\right.\;.
\end{equation}
The next results shows that the families $\{\chi_\gamma\}_{\gamma\in \Gamma_{(\alpha,\beta)}}$ are  \emph{Bessel sequences} \cite[Definition 3.2.2]{christensen-03}
with Bessel bound $M_{\alpha,\beta}$.
This fact is proved in \cite[Theorem (i)]{daubechies-grossmann-88} but we sketch here the argument of the proof.
\begin{prop}\label{prop:bess_fr}
For every $\alpha>0$ and $\beta>0$   one has that
\[
\sum_{\gamma\in\Gamma_{(\alpha,\beta)}}|\langle\chi_\gamma,\phi \rangle|^2\;\leqslant\;M_{\alpha,\beta}\|\phi\|^2\;,\qquad \forall\; \phi\in\;\s{H}_0\;
\] 
where  $M_{\alpha,\beta}$ is given by \eqref{eq:B_bou}.
\end{prop}
\begin{proof}
A direct computation yields that
\[
\sum_{\gamma'\in\Gamma_{(\alpha,\beta)}}|\langle\chi_{\gamma},\chi_{\gamma'} \rangle| =  M_{\alpha,\beta}.
\]
If we define the operator $\bb{Z}$ on $\ell^2(\Gamma_{(\alpha,\beta)})$ (the Gram matrix) by
\begin{equation}
(\bb{Z}c)_\gamma := \sum_{\gamma'\in\Gamma_{(\alpha,\beta)}}\langle\chi_{\gamma},\chi_{\gamma'} \rangle\;c_{\gamma'}, 
\end{equation}
the identity above implies that $\bb{Z} $ is a bounded operator of norm $\|\bb{Z}\|_{\ell^2(\Gamma_{(\alpha,\beta)})}\leqslant M_{\alpha,\beta}$, see~\cite[Lemma 3.5.3]{christensen-03}. It follows that the synthesis operator $J:\ell^2(\Gamma_{(\alpha,\beta)})\to\s{H}_0$ defined by
\begin{equation}\label{eq;synt}
Jc := \sum_{\gamma\in\Gamma_{(\alpha,\beta)}}c_\gamma\;\chi_\gamma\;.
\end{equation}
is bounded since
\begin{equation*}
\Vert Jc\Vert_{\s{H}_0}^2 = \langle c, \bb{Z} c\rangle \leqslant \Vert c\Vert_{\ell^2(\Gamma_{(\alpha,\beta)})} \Vert \bb{Z} c\Vert_{\s{H}_0}\leqslant M_{\alpha,\beta}\Vert c\Vert_{\ell^2(\Gamma_{(\alpha,\beta)})}. 
\end{equation*}
Noting that $(J^*\phi)_\gamma = \langle\chi_\gamma,\phi\rangle$ we have that 
\begin{equation*}
\sum_{\gamma\in\Gamma_{(\alpha,\beta)}}|\langle\chi_\gamma,\phi \rangle|^2 = \Vert J^* \phi\Vert_{\ell^2(\Gamma_{(\alpha,\beta)})}^2 \leqslant \Vert J\Vert \Vert \phi\Vert_{\ell^2(\Gamma_{(\alpha,\beta)})},
\end{equation*}
which concludes the proof.
\end{proof}

\begin{rem}[Injectivity and symmetry of the frame operator]\label{Rk_inj_S}
(i) By the above, the operator $S$ of~\eqref{eq:fr_op} is well-defined and bounded. It is injective if and only if the associated family is complete for  $\s{H}_0$: By its very definition, it follows that
\[
\langle\phi,S\phi\rangle = \sum_{\gamma\in\Gamma_{(\alpha,\beta)}}|\langle\chi_\gamma,\phi \rangle|^2,
\]
and so $S\phi=0$ if and only if $\langle\chi_\gamma,\phi \rangle=0$ for every $\gamma\in\Gamma_{(\alpha,\beta)}$. As a result, Theorem \ref{th:LLL01} implies that that $S$ is not injective if  $\Delta_{\alpha,\beta}> 2\pi$. \\
\noindent (ii) The operator $S$ commutes with the magnetic translations: $[T_{\ell_B^{-1}\gamma},S]=0$ for every $\gamma\in \Gamma_{(\alpha,\beta)}$. Indeed, a direct computation shows that
\begin{align*}
(T_{\ell_B^{-1}\gamma} S T_{\ell_B^{-1}\gamma}^{-1})\phi\;&=\;T_{\ell_B^{-1}\gamma}\left(\sum_{\gamma'\in\Gamma_{(\alpha,\beta)}}\langle\chi_{\gamma'},T_{\ell_B^{-1}\gamma}^{-1}\phi \rangle\;\chi_{\gamma'}\right)\\
&=\;\sum_{\gamma'\in\Gamma_{(\alpha,\beta)}}\langle T_{\ell_B^{-1}\gamma} \chi_{\gamma'},\phi \rangle\;T_{\ell_B^{-1}\gamma} \chi_{\gamma'}=\;\sum_{\gamma'\in\Gamma_{(\alpha,\beta)}}\langle\chi_{\gamma'+\gamma},\phi \rangle\;\chi_{\gamma'+\gamma} = S\phi\;
\end{align*}
for every $\phi\in\s{H}_0$.
\hfill $\blacktriangleleft$
\end{rem}
%


Due to the last remark from now on we will only deal with the case $\Delta_{\alpha,\beta}\leqslant 2\pi$. The fact that if  $\Delta_{\alpha,\beta}<2\pi$ then the family  $\{\chi_\gamma\}_{\gamma\in\Gamma_{(\alpha,\beta)}}$ is a frame as proved in~\cite{daubechies-90,janssen-94}. 
\begin{thm}\label{teo:frame<}
Let $\Delta_{\alpha,\beta}<2\pi$. Then $\{\chi_\gamma\}_{\gamma\in\Gamma_{(\alpha,\beta)}}$ is a {frame} for the $\s{H}_0$.
\end{thm}
\begin{proof}
It suffices to exhibit a constant $m_{\alpha,\beta}>0$ such that 
\begin{equation}\label{eq:frem_0XX1}
{m_{\alpha,\beta}}\;\|\phi\|^2\;\leqslant\; \sum_{\gamma\in\Gamma_{(\alpha,\beta)}}|\langle\chi_\gamma,\phi \rangle|^2\;\leqslant\;{M_{\alpha,\beta}}\;\|\phi\|^2\;
\end{equation}
for every $\phi\in\s{H}_0$; The upper bound is guaranteed by 
Proposition \ref{prop:bess_fr}. Let $\varepsilon:=2\pi-\Delta_{\alpha,\beta}>0$. There exists a family $\{\eta^\varepsilon_\gamma\}_{\gamma\in\Gamma_{(\alpha,\beta)}}\subset \s{H}_0$ (a so-called dual frame) such that (i) there exists a $A_{\varepsilon}>0$ such that
\[
\sum_{\gamma\in\Gamma_{(\alpha,\beta)}}|\langle\eta^\varepsilon_\gamma,\phi \rangle|^2\;\leqslant\;{A_{\varepsilon}}\;\|\phi\|^2
\]
for all $\phi\in\s{H}_0$, and (ii)
\[
\sum_{\gamma\in\Gamma_{(\alpha,\beta)}}\langle\varphi,\eta^\varepsilon_\gamma  \rangle\langle\chi_\gamma,\phi \rangle = \langle\varphi,\phi \rangle
\]
for every $\varphi,\phi\in\s{H}_0$, see~\cite[Proposition B]{janssen-94}. Picking $\varphi=\phi$ in (ii), one gets
\[
\|\phi\|^4\;\leqslant\;\left(\sum_{\gamma\in\Gamma_{(\alpha,\beta)}}|\langle\eta^\varepsilon_\gamma,\phi \rangle|^2\right)\left(\sum_{\gamma\in\Gamma_{(\alpha,\beta)}}|\langle\chi_\gamma,\phi \rangle|^2\right)\;.
\]
by Cauchy-Schwarz. With this, (i) implies the desired inequality with $m_{\alpha,\beta}=A_{\varepsilon}^{-1}$.
\end{proof}


Finally let us briefly comment on the threshold case $\Delta_{\alpha,\beta}=2\pi$. Its singular nature can be observed in the following result:
 
 \begin{prop}\label{prop:specS_tresh}
For  $\Delta_{\alpha,\beta}=2\pi$ the 
frame operator $S$ associated family is a $\{\chi_\gamma\}_{\gamma\in\Gamma_{(\alpha,\beta)}}$   is injective and
\begin{equation}\label{eq_inf_spec}
\min \sigma(S) = 0\;.
\end{equation}
 Therefore $S^{-1}$ exists as a densely defined unbounded operator on $\s{H}_0$.
 \end{prop}
 
\noindent In particular, one can still define a decomposition of the type
 \[
 \phi\;\simeq\;\sum_{\gamma\in\Gamma_{\vartheta}}\langle S^{-1}\chi_\gamma,\phi \rangle\;\chi_\gamma\;.
 \]
 for $\phi\in\s{H}_0$. However, the sign $\simeq$ means that the convergence is not meant in the sense of the $\ell^2$-convergence of the frame coefficients $\langle S^{-1}\chi_\gamma,\phi \rangle$ in view of the fact that $S^{-1}$ is unbounded. This is known as the unstable case, see for example~\cite[Section II - C - Examples i) ]{daubechies-90} and references therein.

\subsection{Magnetic Gabor frames for the higer Landau levels}
The Landau level of order $r\in\N_0$ (rLL) is defined by
\[
\s{H}_r = \overline{{\rm span}\;\{\psi_{(r,m)}\;|\; m\in\N_0\}}\;.
\]
Since from \eqref{eq:herm2} one has that
\begin{equation}\label{eq;psi_r}
\psi_{(r,m)}(x) = \frac{1}{\sqrt{r!}}(\rr{a}^+)^r\psi_{(0,m)}
\end{equation}
one gets from the expansion \eqref{eq;chi_002} that
\begin{equation}\label{eq;chi_002-XX}
\chi_{(r,\gamma)} := \frac{1}{\sqrt{r!}}(\rr{a}^+)^r\chi_\gamma = \expo{-\frac{|\gamma|^2}{4\ell_B^2}}\sum_{m\in\N_0}\frac{1}{\sqrt{m!}}\left(\frac{\gamma_1+\ii \gamma_2}{\ell_B\sqrt{2}}\right)^m\psi_{(r,m)}
\end{equation}
with the notation  $\chi_{\gamma}\equiv\chi_{(0,\gamma)}$.
Therefore $\chi_{(r,\gamma)}\in \s{H}_r$ for every $\gamma\in\R^2$. 
Equation \eqref{eq;psi_r} immediately implies that the map $(r!)^{-1/2}(\rr{a}^+)^r$ is a unitary equivalence from $\s{H}_0$ to $\s{H}_r$. Therefore equation \eqref{eq;chi_002-XX} implies that, for any $r\in\bbN$, the family $\{\chi_{(r,\gamma)}\}_{\gamma\in\Gamma_{(\alpha,\beta)}}$ shares the same frame property with respect to $\s{H}_r$ as the family $\{\chi_{\gamma}\}_{\gamma\in\Gamma_{(\alpha,\beta)}}$ with respect to $\s{H}_0$; In particular, (\ref{eq:gaus_orth}) holds uniformly for all Landau levels.

\subsection{Interacting dynamics}
The results of this section, together with Theorem~\ref{teo:frame<} and Proposition~\ref{prop:01} prove that if $\Delta_{\alpha,\beta}<2\pi$, then the family $\{\chi_{(r,\gamma)}\}_{r\in\mathbb{N}_0,\gamma\in\Gamma_{(\alpha,\beta)}}$ is a bonafide lattice-localized frame in the sense of Section~\ref{sec: Gabor frames}. In fact, we have a concrete realization of the setting of Example~\ref{Rk_exp_gab_2}. Hence the results of Section~\ref{sec: LRB} apply and allow for the definition of good interacting quantum dynamics. The allowed interactions are \emph{local in space} and can \emph{mix of the various Landau levels}, although here again only in a local-in-energy fashion.


\section{Additional proofs}\label{sec:add_proof}

\subsection{The proof of Proposition~\ref{Prop: localization of S^-p}}
In this section we provide a complete proof of Proposition~\ref{Prop: localization of S^-p} which turns out to be somehow technical. The proof is inspired by \cite{jaffard-90,pilipovic-prangoski-zigic-19}: We first obtain estimates on the matrix elements of $S^p$ with $p\in\N$, and use them and the Neuman series to prove the main bound~(\ref{matrix elements of S-p}). For any $\gamma,\gamma'\in\Gamma$, let
\[
\bb{S}_{\gamma,\gamma'}^{(p)} := \langle\chi_\gamma, S^p\chi_{\gamma'} \rangle\;,
\]
and we shall consider the associated infinite matrix
\begin{equation}\label{def scr Sp}
\bb{S}^{(p)} := \left\{\bb{S}_{\gamma,\gamma'}^{(p)}\right\}_{\gamma,\gamma'\in\Gamma}\;.
\end{equation}
\begin{defn}
Let  $\rr{M}_\lambda$ be the  space of the matrices $\bb{A}=\left\{\bb{A}_{\gamma,\gamma'}\right\}_{\gamma,\gamma'\in\Gamma}$, indexed by $\Gamma\times\Gamma$, such that for all $0<\delta<\lambda$ there exists a $c_\delta>0$ such that
\[
|\bb{A}_{\gamma,\gamma'}|\;\leqslant\;c_\delta\;\expo{-(\lambda-\delta) d(\gamma,\gamma')}\;,\qquad \forall\;\gamma,\gamma'\in\Gamma\;.
\]
\end{defn}
\begin{rem}\label{rk_sat_lamb}
If $\lambda'<\lambda$, then $\rr{M}_\lambda\subset \rr{M}_{\lambda'}$.
\hfill $\blacktriangleleft$
\end{rem}
%


Let $\bb{A}=\left\{\bb{A}_{\gamma,\gamma'}\right\}_{\gamma,\gamma'\in\Gamma}$ and $\bb{B}=\left\{\bb{B}_{\gamma,\gamma'}\right\}_{\gamma,\gamma'\in\Gamma}$ two matrices. We will denote by $\bb{A}\cdot\bb{B}$ the matrix with entries given by
 \[
 (\bb{A}\cdot\bb{B})_{\gamma,\gamma'} := \sum_{\xi\in\Gamma}\bb{A}_{\gamma,\xi}\bb{B}_{\xi,\gamma'}\;.
 \]

\begin{lemma}\label{lemma:algeb_M}
The space $\rr{M}_\lambda$ is an algebra over $\C$.
\end{lemma}
\begin{proof}
Let $\bb{A},\bb{B}\in\rr{M}_\lambda$. The fact that $a\bb{A}+b\bb{B}\in\rr{M}_\lambda$ for every $a,b\in\C$ is  straightforward. Let  $0<\delta'<\delta<\lambda$. There is $C>0$ such that
observe that
\begin{align*}
\left|\sum_{\xi\in\Gamma}\bb{A}_{\gamma,\xi}\bb{B}_{\xi,\gamma'}\right|\;&\leqslant\;C\sum_{\xi\in\Gamma}
\expo{-(\lambda-\delta') d(\gamma,\xi)} \expo{-(\lambda-\delta) d(\xi,\gamma')}\\
&=\;C\sum_{\xi\in\Gamma}
\expo{-(\delta-\delta')d(\gamma,\xi)}\expo{- (\lambda-\delta)[ d(\gamma,\xi)+ d(\xi,\gamma')]} \leqslant\; c_{\delta',\delta} \expo{- (\lambda-\delta) d(\gamma,\gamma')}
\end{align*}
where
\begin{equation}\label{eq:const_cp}
c_{\delta',\delta} := C\bigg(\sup_{\gamma\in\Gamma}\sum_{\xi\in\Gamma}\expo{-(\delta-\delta')d(\gamma,\xi)}\bigg)
\end{equation}
is  finite in view of the assumed bound~\eqref{eq:bound_cond}. Since $0<\delta<\lambda$ is arbitrary, this shows that $\bb{A}\cdot\bb{B}$ is a well-defined element of $\rr{M}_\lambda$ indeed. 
\end{proof}

\begin{lemma}\label{prop:estiSp0}
Let $\bb{S}^{(p)}$ be defined by~(\ref{def scr Sp}). Then $\bb{S}^{(p)}\in \rr{M}_{\lambda}$ for every $p\in\N_0$ and  
one has that
\begin{equation}\label{eq:best_ineq_S_P}
|\langle\chi_\gamma, S^p\chi_{\gamma'} \rangle|\;\leqslant\; Gc^p_\epsilon
\expo{- (\lambda-\delta) d(\gamma,\gamma')}\;,\qquad \forall\;\gamma,\gamma'\in\Gamma\;.
\end{equation}
for every $0<\epsilon <\delta<\lambda$ 
where $
c_\epsilon := G m_\epsilon$ 
and $m_\epsilon$ is given by \eqref{eq:bound_cond}.
\end{lemma}
\begin{proof}
From Definition \ref{def:loc_fr} (and Remark \ref{rk_sat_lamb}) it follows that $\bb{S}^{(0)}\in \rr{M}_{\lambda}$ and \eqref{eq:best_ineq_S_P} holds for $p=0$. By definition of the frame operator,  
\[
\langle\chi_\gamma, S\chi_{\gamma'} \rangle = 
\sum_{\xi\in\Gamma}\langle\chi_\xi,\chi_{\gamma'} \rangle\;\langle\chi_\gamma, \chi_\xi\rangle\;
\]
which can be phrased as $\bb{S}_{\gamma,\gamma'}^{(1)}=(\bb{S}^{(0)}\cdot\bb{S}^{(0)})_{\gamma,\gamma'}$. In view of Lemma \ref{lemma:algeb_M} this implies that also $\bb{S}^{(1)}\in \rr{M}_{\lambda}$. The proof of Lemma \ref{lemma:algeb_M} with $C=G^2$ and $\delta':=\delta-\epsilon>0$ one obtains that the bound~\eqref{eq:best_ineq_S_P} is satisfied also for $p=1$. 
The general case follows similarly by induction. First of all,
\begin{equation}\label{eq:matSp}
\langle\chi_\gamma, S^p\chi_{\gamma'} \rangle = 
\sum_{\xi\in\Gamma_{(\alpha,\beta)}}\langle\chi_\xi,S^{p-1}\chi_{\gamma'} \rangle\;\langle\chi_\gamma, \chi_\xi\rangle\;
\end{equation}
implies that $\bb{S}_{\gamma,\gamma'}^{(p)} = (\bb{S}^{(0)}\cdot\bb{S}^{(p-1)})_{\gamma,\gamma'}$. Therefore, by applying inductively  Lemma \ref{lemma:algeb_M}
one gets that $\bb{S}^{(p)}\in \rr{M}_{\lambda}$. Moreover, by the induction hypothesis, one has that
\[
\left|\bb{S}_{\gamma,\gamma'}^{(p)}\right|\;\leqslant\;G c_\epsilon^{p-1}
\sum_{\xi\in\Gamma}
\expo{-(\lambda-\delta') d(\gamma,\xi)} \expo{-(\lambda-\delta) d(\xi,\gamma')}
\]
with $0<\delta'<\delta<\lambda$. We conclude as above.
\end{proof}

\begin{rem}
With the notation $\bb{Z}:=\bb{S}^{(0)}$, we have that $\bb{S}^{(p)}=\bb{Z}\cdot\ldots\cdot\bb{Z}=\bb{Z}^{p+1}
$ for all $p\in\N_0$.
On the other hand, from
\[
\psi = S(S^{-1}\psi) = \sum_{\xi\in\Gamma}\langle\chi_\xi,S^{-1}\psi\rangle\chi_\xi
\]
one gets
\[
\langle\chi_\gamma,\chi_{\gamma'}\rangle = \sum_{\xi\in\Gamma}\langle\chi_\xi,S^{-1}\chi_{\gamma'}\rangle\langle\chi_\gamma,\chi_{\xi}\rangle\;
\]
and similarly from 
\[
\psi = S^{-1}(S\psi) = \sum_{\xi\in\Gamma}\langle\chi_\xi,\psi\rangle S^{-1}\chi_\xi
\]
one gets
\[
\langle\chi_\gamma,\chi_{\gamma'}\rangle = \sum_{\xi\in\Gamma}\langle\chi_\xi,\chi_{\gamma'}\rangle\langle\chi_\gamma,S^{-1}\chi_{\xi}\rangle\;.
\]
In summary one obtains
\[
\bb{Z} = \bb{Z}\cdot\bb{S}^{(-1)} = \bb{S}^{(-1)}\cdot\bb{Z}\;.
\]
By induction one can prove that
\[
\bb{Z} = \bb{Z}^p\cdot\bb{S}^{(-p)} = \bb{S}^{(-p)}\cdot\bb{Z}^p\;,\qquad p\in\N\;.
\]
We also note that $\bb{S}^{(-p)}\neq(\bb{S}^{(p)})^{-1}$, namely  the matrix $\bb{S}^{(-p)}$ is of course not the inverse of the matrix $\bb{S}^{(p)}$.
\hfill $\blacktriangleleft$
\end{rem}

We are now in position to provide the proof of Proposition~\ref{Prop: localization of S^-p}.

\begin{proof}[Proof of Proposition~\ref{Prop: localization of S^-p}.]
Let $R_p:={\bf 1}- S^p\|S\|^{-p}$.
Since $S$ is positive and invertible, one has that
\[
r_p := \|R_p\| = \left\|{\bf 1}-\frac{S^p}{\|S\|^p}\right\| = 1-\left(\frac{s}{\|S\|}\right)^p\;<\;1\;
\] 
 with $s:=\min \sigma(S)>0$.
Therefore,
\[
S^{-p} = (S^p)^{-1} = \frac{1}{\|S\|^p}\sum_{k=0}^{+\infty}R_p^k\;.
\]
We expand $R_p^k$ in a binomial sum and Lemma~\ref{prop:estiSp0} yields
\begin{align*}
|\langle\chi_\gamma, R_p^k\chi_{\gamma'} \rangle|\;
&\leqslant\;G\sum_{j=0}^k\binom{k}{j}
\left(\frac{c_\epsilon^p}{\|S\|^p}\right)^j\expo{- (\lambda-\delta) d(\gamma,\gamma')}
\\
&=\;G\left(1+\frac{c_\epsilon^p}{\|S\|^p}\right)^k\expo{- (\lambda-\delta) d(\gamma,\gamma')}\;
\end{align*}
for every $0<\epsilon <\delta<\lambda$. On the other hand, it is also true that
\[
|\langle\chi_\gamma, R_p^k\chi_{\gamma'} \rangle|\;\leqslant\;r_p^k\;<\;1
\]
for all $k$ and so
\[
|\langle\chi_\gamma, R_p^k\chi_{\gamma'} \rangle|\;\leqslant\;\min\left\{r_p^k, d_{p,\epsilon}^k \expo{- (\lambda-\delta) d(\gamma,\gamma')}\right\}
\]
where $d_{p,\epsilon}:=G(1+\frac{c_\epsilon^p}{\|S\|^p})>1$, since $L\geqslant1$. As a consequence
\[
|\langle\chi_\gamma, S^{-p}\chi_{\gamma'} \rangle|\;\leqslant\;\frac{1}{\|S\|^p}
\sum_{k=0}^{+\infty}\min\left\{r_p^k, d_{p,\epsilon}^k \expo{- (\lambda-\delta) d(\gamma,\gamma')}\right\}\;.
\]
We shall now find an optimal $k_*$ above which $r_p^k$ will be used over the second factor. Let us fix a $0<\theta<\lambda-\delta$ and
for each $\gamma,\gamma'\in\Gamma$, let $k_\ast\equiv k_\ast(\gamma,\gamma',\theta)\in\N$ be the largest integer such that
\[
\left(\frac{d_{p,\epsilon}}{r_p}\right)^k\;\leqslant\;\expo{ (\lambda-\delta-\theta) d(\gamma,\gamma')}\;<\;\expo{ (\lambda-\delta) d(\gamma,\gamma')}\;,\qquad \forall\; k\leqslant k_\ast\;.
\]
Such a $k_\ast$ is explicitly given by $k_\ast = \lfloor E_{p,\theta} d(\gamma,\gamma') \rfloor$ where $E_{p,\theta} := \frac{(\lambda-\delta-\theta) }{\ln\left(\frac{d_{p,\epsilon}}{r_p}\right)}$. Since $E_{p,\theta}d(\gamma,\gamma')\leqslant k_\ast+1$ one concludes that
\[
r_p^{k_\ast+1}\;\leqslant\;r_p^{E_{p,\theta}d(\gamma,\gamma')} = \expo{\ln(r_p)[E_{p,\theta}d(\gamma,\gamma')]}\;.
\]
From this, one infers that 
\[
\begin{aligned}
|\langle\chi_\gamma, S^{-p}\chi_{\gamma'} \rangle|\;&\leqslant\;\frac{1}{\|S\|^p}\left(\sum_{k=0}^{k_\ast}d_{p,\epsilon}^k \expo{- (\lambda-\delta) d(\gamma,\gamma')}\;+\;\sum_{k=k_\ast+1}^{\infty}r_p^k\right)\\
&=\;\frac{1}{\|S\|^p}\left(\expo{- \theta d(\gamma,\gamma')}\sum_{k=0}^{k_\ast}d_{p,\epsilon}^k \expo{- (\lambda-\delta-\theta) d(\gamma,\gamma')}\;+\;r_p^{k_\ast+1}\sum_{k=0}^{\infty}r_p^k\right)\\
&\leqslant\;\frac{1}{\|S\|^p}\left(\expo{- \theta d(\gamma,\gamma')}\sum_{k=0}^{k_\ast}r_p^k\;+\;\expo{\ln(r_p)E_{p,\theta}d(\gamma,\gamma')}\sum_{k=0}^{\infty}r_p^k\right)\\
&\leqslant\;\frac{\expo{- \theta d(\gamma,\gamma')}+\expo{\ln(r_p)E_{p,\theta}d(\gamma,\gamma')}}{\|S\|^p(1-r_p)}\;.
\end{aligned}
\]
The claim follows with the choices
\[
\lambda_p := \min\left\{\theta, \ln\left(\frac{1}{r_p}\right)E_{p,\theta}\right\}\quad\text{and}\quad a_p := \frac{2}{\|S\|^p(1-r_p)}.
\]
Note that the choice of $k_\ast$ depended on $\gamma,\gamma'$, but none of the constants above. 
\end{proof}


The following is just a restatement of Proposition~\ref{Prop: localization of S^-p}.

\begin{cor}
For every $p\in\N$ there exists a $0<\lambda_p<\lambda$   such that $\bb{S}^{(-p)}\in\rr{M}_{\lambda_p}$.
\end{cor}

\begin{ex}\label{Rk_exp_gab_exp_cos}
We provide explicit formulas for the constants $\lambda_p$ and $a_p$ for the lattice of coherent states described in Section \ref{sec:MagGabor} and Example \ref{Rk_exp_gab}.
In this situation  by adapting the computation of Lemma \ref{lemma:algeb_M} with $\delta'=0$ and $\delta=\frac{\lambda}{2}$ one gets
\[
|\langle\chi_\gamma, S\chi_{\gamma'} \rangle|\;\leqslant\; \left(\sum_{\xi\in\Gamma}\expo{- \frac{\lambda}{2}d(\xi,\gamma')}\right)
\expo{- \frac{\lambda}{2}d(\gamma,\gamma')}\;,\qquad \forall\;\gamma,\gamma'\in\Gamma\;.
\]
Since
\[
c_\ast := \left(\frac{2}{1-\expo{-\frac{\varpi}{2}\alpha_\ast^3}}\right)^d\;\geqslant\;\sum_{\xi\in\Gamma}\expo{- \frac{\lambda}{2}d(\xi,\gamma')}
\]
as computed in Example \ref{Rk_exp_gab}, one obtains that 
\[
|\langle\chi_\gamma, S^p\chi_{\gamma'} \rangle|\;\leqslant\; c_\ast^p
\expo{- \frac{\lambda}{2}d(\gamma,\gamma')}\;,\qquad \forall\;\gamma,\gamma'\in\Gamma\;
\]
by mimicking the same argument  of Lemma~\ref{prop:estiSp0}.
Now let us define $d_{p}:=1+(\frac{c_\ast}{\|S\|})^p$ (in this case $G=1$) and $r_p:=1-(\frac{s}{\|S\|})^p$ as in the proof of  Proposition~\ref{Prop: localization of S^-p}. 
Since $\frac{d_p}{r_p}>\frac{1}{r_p}>1$, it follows that 
\[
 0\;<\;\frac{\ln\left(\frac{1}{r_p}\right)}{\ln\left(\frac{d_p}{r_p}\right)}\;<\;1\;.
\]
By choosing $\theta_m:=\frac{\lambda}{2m}$ for some $m\in\N$, let us define
\[
E_p := \frac{(\frac{\lambda}{2}-\theta_m) }{\ln\left(\frac{d_{p}}{r_p}\right)} = 
\frac{\lambda}{2}\left(1-\frac{1}{m}\right)\frac{1 }{\ln\left(\frac{d_{p}}{r_p}\right)}\;.
\]
The condition $\theta_m\leqslant -\ln(r_p)E_p$
translates in 
\[
\frac{1}{m}\;\leqslant\;\left(1-\frac{1}{m}\right)\frac{\ln\left(\frac{1}{r_p}\right) }{\ln\left(\frac{d_{p}}{r_p}\right)}
\]
which implies
\[
m\;\geqslant\;1+\frac{\ln\left(\frac{d_{p}}{r_p}\right)}{\ln\left(\frac{1}{r_p}\right) }\;>\;2\;.
\]
By introducing
\[
m_p := 2+\left\lfloor\frac{\ln\left(\frac{d_{p}}{r_p}\right)}{\ln\left(\frac{1}{r_p}\right) }\right	\rfloor
\]
one can pick
\[
\lambda_p\::=\:\frac{\lambda}{2m_p}\;<\;\frac{\lambda}{4}
\]
for the exponent in Proposition~\ref{Prop: localization of S^-p}.\end{ex}

\section*{Conflicts of interest statement}

The authors certify that they have NO affiliations with or involvement in any organization or entity with any financial interest (such as honoraria; educational grants; participation in speakers’ bureaus; membership, employment, consultancies, stock ownership, or other equity interest; and expert testimony or patent-licensing arrangements), or non-financial interest (such as personal or professional relationships, affiliations, knowledge or beliefs) in the subject matter or materials discussed in this manuscript.

\end{document}